\documentclass[12pt]{article}
\usepackage[latin1]{inputenc}
\usepackage[british]{babel}
\usepackage{cmap}
\usepackage{lmodern}

\usepackage{amssymb, amsmath, amsthm}
\usepackage[a4paper,top=25mm,bottom=25mm,left=25mm,right=25mm]{geometry}
\usepackage{etex}
\usepackage{ragged2e}

\usepackage{authblk} 
\usepackage{pifont}
\usepackage{graphicx}
\usepackage[usenames,dvipsnames,svgnames,table]{xcolor}
\usepackage[figuresright]{rotating}
\usepackage{xtab} 
\usepackage{longtable} 
\usepackage{multirow}
\usepackage{footnote}
\usepackage[stable]{footmisc}
\usepackage{chngpage} 
\usepackage{pdflscape} 
\usepackage[nottoc,notlot,notlof]{tocbibind} 

\usepackage{pgfplots}
\pgfplotsset{compat=1.14}
\usepackage{setspace}

\makesavenoteenv{tabular}
\usepackage{tabularx}
\usepackage{booktabs}
\usepackage{threeparttable}
\usepackage[referable]{threeparttablex} 
\newcolumntype{R}{>{\raggedleft\arraybackslash}X}
\newcolumntype{L}{>{\raggedright\arraybackslash}X}
\newcolumntype{C}{>{\centering\arraybackslash}X}
\newcolumntype{A}{>{\columncolor{gray!25}}C}
\newcolumntype{a}{>{\columncolor{gray!25}}c}

\newlength{\tablen}

\usepackage{dcolumn} 
\newcolumntype{.}{D{.}{.}{-1}}

\usepackage{tikz}
\usetikzlibrary{arrows, calc, matrix, patterns, positioning, trees}
\usepackage[semicolon]{natbib}
\usepackage[hyphens]{url}
\usepackage{hyperref} 
\hypersetup{
  colorlinks   = true,    
  urlcolor     = blue,    
  linkcolor    = blue,    
  citecolor    = ForestGreen      
}
\usepackage{microtype}
\usepackage[justification=centerfirst]{caption} 

\usepackage[labelformat=simple]{subcaption}

\DeclareCaptionLabelFormat{parenthesis}{(#2)}
\makeatletter
\renewcommand\p@subfigure{\arabic{figure}.}
\makeatother

\DeclareCaptionLabelFormat{parenthesis}{(#2)}
\makeatletter
\renewcommand\p@subtable{\arabic{table}.}
\makeatother

\usepackage{enumitem}
\setlist[itemize]{leftmargin=2.5\parindent}
\setlist[enumerate]{leftmargin=2.5\parindent}

\theoremstyle{plain}

\newtheorem{corollary}{Corollary}[section]

\newtheorem{proposition}{Proposition}[section]
\newtheorem{theorem}{Theorem}[section]

\theoremstyle{definition}

\newtheorem{definition}{Definition}[section]
\newtheorem{example}{Example}[section]

\theoremstyle{remark}

\newtheorem{remark}{Remark}[section]

\def\keywords{\vspace{.5em} 
{\noindent \textit{Keywords}: }}

\def\JEL{\vspace{.5em} 
{\noindent \textbf{\emph{JEL} classification number}: }}

\def\AMS{\vspace{.5em} 
{\noindent \textbf{\emph{MSC} class}: }}

\author{\href{https://sites.google.com/view/laszlocsato}{L\'aszl\'o Csat\'o}\thanks{~E-mail: laszlo.csato@sztaki.hu} }
\affil{Institute for Computer Science and Control (SZTAKI) \\
Laboratory on Engineering and Management Intelligence, Research Group of Operations Research and Decision Systems}
\affil{Corvinus University of Budapest (BCE) \\
Department of Operations Research and Actuarial Sciences}
\affil{Budapest, Hungary}
\title{The incentive (in)compatibility of \\ group-based qualification systems}
\date{\today}

\def\Dedication{ 
{\noindent $\mathfrak{Ein}$ $\mathfrak{Fehler}$ $\mathfrak{in}$ $\mathfrak{der}$ $\mathfrak{urspr \ddot{u} nglichen}$ $\mathfrak{Versammlung}$ $\mathfrak{der}$ $\mathfrak{Heere}$ $\mathfrak{ist}$ \\
$\mathfrak{im}$ $\mathfrak{ganzen}$ $\mathfrak{Verlauf}$ $\mathfrak{des}$ $\mathfrak{Feldzuges}$ $\mathfrak{kaum}$ $\mathfrak{wieder}$ $\mathfrak{gut}$ $\mathfrak{zu}$ $\mathfrak{machen}$.\footnote{~``\emph{A single error in the original assembly of the armies can hardly ever be rectified during the entire course of the campaign.}'' (Source: Holger H. Herwig: \emph{The Forgotten Campaign: Alsace-Lorraine August 1914}, p.~195, in Michael S. Neiberg (ed.): \emph{Arms and the Man: Military History Essays in Honor of Dennis Showalter}, Brill Academic Publishers, Leiden, Boston, 2011.)}
}

\vspace{0.25cm}
\noindent (Helmuth Karl Bernhard von Moltke: \emph{Taktisch-strategische \\ Aufs\"atze aus den Jahren 1857 bis 1871})
\vspace{1cm} }

\begin{document}
\newgeometry{top=20mm,bottom=25mm,left=25mm,right=25mm}

\maketitle
\thispagestyle{empty}
\Dedication

\begin{abstract}
\noindent
Tournament organisers supposedly design rules such that a team cannot be strictly better off by exerting a lower effort. However, the European qualification tournaments for recent FIFA soccer World Cups are known to violate this requirement, which inspires our study on the incentive compatibility of similar group-based qualification systems.
Theorems listing the sufficient and necessary conditions of strategy-proofness are provided and applied to classify several soccer qualification tournaments for FIFA World Cups and UEFA European Championships.
Two reasonable mechanisms are proposed to solve the problem of incentive incompatibility: the first is based on abolishing the anonymity of the matches discarded in the comparison of teams from different groups, while the second involves a rethinking of the seeding procedure.
Our results have useful implications for the governing bodies of major sports.

\JEL{C44, D71, Z20}

\AMS{62F07, 91B14}

\keywords{mechanism design; strategy-proofness; manipulation; tournament ranking; FIFA World Cup}
\end{abstract}

\clearpage
\newgeometry{top=25mm,bottom=25mm,left=25mm,right=25mm}

\section{Introduction} \label{Sec1}

Incentive compatibility is a crucial concept of mechanism design at least since the famous Gibbard-Satterthwaite impossibility theorem \citep{Gibbard1973, Satterthwaite1975}, which states that any fair (non-dictatorial) voting rule is susceptible to strategic voting: there always exists a voter who can achieve a more favourable outcome by being insincere in the presence of more than two alternatives. The violation of strategy-proofness may have unwelcome implications such as a governing coalition of parties losing seats by getting more votes \citep{Tasnadi2008}, or a journal improving its ranking by making additional fake citations to other journals \citep{KoczyStrobel2009}.

This issue can be even more essential in professional sports where contestants should exert costly efforts and the tournaments involve high-stake decisions that are familiar to all agents \citep{Kahn2000}.
Consequently, all sporting contests should provide players with appropriate incentives to perform \citep{Szymanski2003}.

While some scholars may argue that sports are a somewhat frivolous topic compared to most other industries, it is of great interest to a high number of people \citep{Wright2014}. For example, the final of the 2010 FIFA World Cup has been watched by about half of the humans who were alive on its day \citep{Palacios-Huerta2014}.\footnote{~FIFA stands for \emph{F\'ed\'eration Internationale de Football Association}, it is the international governing body of soccer.}
Since sports is also a huge business with billion dollars of turnover, the optimal design of tournaments poses an important theoretical problem for both economics \citep{Szymanski2003} and operations research \citep{ScarfYusofBilbao2009}.

Furthermore, it seems that the analysis of problems from sports can also add something for economists. Penalty kicks in soccer provide an excellent environment to test mixed strategies \citep{ChiapporiLevittGroseclose2002, Palacios-Huerta2003}, while penalty shootouts \citep{ApesteguiaPalacios-Huerta2010, KocherLenzSutter2012, Palacios-Huerta2014, VandebroekMcCannVroom2018, ArrondelDuhautoisLaslier2019}, multi-stage chess contests \citep{Gonzalez-DiazPalacios-Huerta2016}, and tennis tiebreaks \citep{MagnusKlaassen1999, Cohen-ZadaKrumerShapir2018} have inspired research on the effect of the order of actions in sequential contests. Revenue sharing in sports can contribute to the development of  allocation rules \citep{BergantinosMoreno-Ternero2019a}. Somewhat surprisingly, even macroeconomics may benefit from the analysis of sports data, for example, \citet{KrauseSzymanski2019} use the performance of national soccer teams to test convergence across countries and the presence of the ``middle income trap''.

The enterprise of exploring sports rules from the perspective of strategy-proofness is still in its infancy, and much remains to be done.
Incentive incompatibility sometimes originates from the fact that being ranked lower in a given (group) stage leads to facing a more preferred competitor in the subsequent (knock-out) stage, which means an advantage in terms of expected probability. For example, \citet{Pauly2014} develops a mathematical model of such manipulation in round-robin subtournaments and derives an impossibility theorem, while \citet{Vong2017} considers the strategic manipulation problem in multistage tournaments and shows that only the top-ranked player can be allowed to qualify from each group in order to make any equilibrium ranking of qualifying players immune to manipulation.

It is also a well-known fact that sports using player drafts with the traditional set-up of reverse order can make losing a profitable strategy \citep{TaylorTrogdon2002, BalsdonFongThayer2007}. This means a serious problem in the National Basketball Association (NBA), the men's professional basketball league in North America, where such ``tanking'' even influences the betting markets \citep{SoebbingHumphreys2013}.
\citet{Fornwagner2019} gives probably the first evidence that there is a concrete strategy behind losing, while \citet{Lenten2016} and \citet{LentenSmithBoys2018} propose an alternative draft-pick allocation policy and evaluate it via a quasi-natural experiment.

These issues and a plethora of similar potential and actual cases \citep{KendallLenten2017} demonstrate that sports administrators often place insufficient resources towards avoiding incentive incompatibility, although this is against the spirit of the game and a misaligned tournament design can easily lead to outrage among consumers.

Here we focus on a potentially more serious problem when a team is \emph{strictly better off} by losing, that is, the manipulation of a match is completely risk-free.
The first academic paper addressing this issue is \citet{DagaevSonin2018}, which proves that tournament systems, consisting of multiple round-robin and knock-out tournaments with noncumulative prizes, usually violate strategy-proofness in the above sense. Their main result reveals the incentive incompatibility of the qualification for the UEFA Europa League, the second-tier competition of European club soccer, before the 2015/16 season.\footnote{~UEFA stands for \emph{Union of European Football Associations}, it is the administrative body of soccer in Europe, although some member states are (partly) located in Asia.}
A follow-up work \citep{Csato2019b} identified the same problem in the qualification for the UEFA Champions League, the most prestigious club competition in European soccer, across the three seasons played between 2016 and 2019.

In particular, the current paper discusses the incentive (in)compatibility of group-based qualification systems---a fundamentally different setting from the multiple qualifiers analysed in \citet{DagaevSonin2018}.
This design has been used in the qualification tournaments for two prominent soccer competitions, the FIFA World Cup (in the European Zone) and UEFA European Championship, among others. While the violation of strategy-proofness by some recent series has already been recognised \citep{DagaevSonin2013, Csato2017d, Csato2018b}, none of these works provide a mathematical model generalising the mere observation, as well as no strategy-proof mechanisms are suggested.

Thus our standalone contributions can be summarised as follows:
(1) the application of the same monotonic ranking for each group including the repechage group (where some teams playing in different groups are compared) is proved to be insufficient for the strategy-proofness of the whole group-based qualification system unless the set of matches considered in the repechage group is chosen appropriately;
(2) theorems listing the sufficient and necessary conditions of incentive compatibility are provided and applied to identify nine real-world qualification systems were susceptible to manipulation;
(3) two different mechanisms are proposed for solving the problem of misaligned incentives.

An intuitive interpretation of the main message is that ranking a set of teams in different pools on a secondary basis along a common denominator is challenging if the size of those pools is not uniform.\footnote{~A related problem is when not the size but the strength of the pools is different, which is not addressed here as being a softer criterion. See \citet{Csato2017c} for a discussion of how to take the varying strength of opponents into account in the ranking of Swiss-system tournaments that are widely used in chess.} In finding a way around it, the criteria (each criterion and their order) for ranking some teams in the secondary sense must be identical to that used in the primary sense, in the pools themselves. Otherwise, manipulation remains possible in certain situations.


The rest of the paper is organised in the following way.
Section~\ref{Sec2} describes the motivating real-world observation, the European qualification tournament for the 2018 FIFA World Cup, and outlines an incentive incompatible scenario in this framework. 
Section~\ref{Sec3} builds and analyses the formal mathematical model of this design, which is applied to examine the strategy-proofness of some recent qualification tournaments from soccer in Section~\ref{Sec4}. 
Section~\ref{Sec5} recommends incentive compatible mechanisms, while Section~\ref{Sec6} concludes.

\section{A real-world example: the 2018 FIFA World Cup qualification tournament (UEFA)} \label{Sec2}

The \href{https://en.wikipedia.org/wiki/2018_FIFA_World_Cup_qualification_(UEFA)}{2018 FIFA World Cup qualification tournament (UEFA)} was the European section of the \href{https://en.wikipedia.org/wiki/2018_FIFA_World_Cup}{2018 FIFA World Cup} qualifiers of national soccer teams, which are members of UEFA.
Russia automatically qualified as a host, therefore---after Gibraltar and Kosovo became FIFA members in May 2016---$54$ teams entered the qualification process for the $13$ slots available in the final tournament.

The format of the competition, confirmed by the UEFA Executive Committee meeting on 22-23 March 2015 in Vienna, is the following:
\begin{itemize}
\item
Group stage (first round): Nine groups of six teams each, playing home-and-away round-robin matches. The winners of each group qualify for the 2018 FIFA World Cup, and the eight best runners-up advance to the play-offs (second round).
\item
Play-offs (second round): The eight best second-placed teams from the group stage
are paired and each pair plays home-and-away matches over two legs. The four winners qualify for the 2018 FIFA World Cup.
\end{itemize}

We focus on the first round, where the tie-breaking rules are \citep[Article~20.6]{FIFA2016}:
(1) greatest number of points obtained in all group matches (with three points for a win, one point for a draw and no points for a defeat);
(2) goal difference in all group matches; and
(3) greatest number of goals scored in all group matches.\footnote{~Strangely, it is not described explicitly that greater goal difference is preferred.}
Further tie-breaking rules will play no role in our discussion.

The choice of the eight best second-placed teams is not addressed in \citet{FIFA2016}, and we were not able to find the relevant regulation. However, according to a FIFA Media Release \citep{FIFA2017a}, which is reinforced by an earlier UEFA press release \citep{UEFA2016b}: ``[\dots] \emph{the eight best runners-up will be decided by ranking criteria as stated in the 2018 FIFA World Cup Regulations, namely points, goal difference, goals scored, goals scored away from home and disciplinary ranking, with the results against teams ranked 6th not being taken into account.}''
\citet{AFC2015} illustrates the ranking of second-placed teams when some group matches are discarded.


Thus the ranking of runners-up strictly follows tie-breaking in groups, with the crucial difference of discarding two matches played against the last-placed team of the group.
It turns out that this, seemingly minor, modification has serious unintended consequences because it allows for possible manipulation of the European qualifiers to the 2018 FIFA World Cup.

As it was mentioned in the Introduction, our starting observation has some history in academic research.
The misaligned incentives have probably been revealed first by a column in the case of the European qualification tournament for the 2014 FIFA World Cup \citep{DagaevSonin2013}.
\citet[p.~1143]{DagaevSonin2018} states that: ``\emph{Still, the problem of misaligned incentives is not restricted to national tournaments. For example, competition rules of the European qualification tournament for the 2014 FIFA World Cup in Brazil suffered from the same problem. And the ``perverse incentives'' situation was not merely a theoretical possibility. Two months before the end of the tournament, with 80\% of games completed, there still was a scenario under which a team might need to achieve a draw instead of winning to go to Brazil.}''
According to \citet{Csato2018b}, France had an incentive to kick two own goals on its last match against Israel in the UEFA Euro 1996 qualifying tournament.

\begin{example} \label{Examp21}
There was a situation with a positive probability in the 2018 FIFA World Cup qualification tournament (UEFA), after four-fifths of all matches were over, under which Bulgaria might need to play a draw instead of winning against Luxembourg on the last matchday of 10 October 2017.

We generate hypothetical results for the last two matchdays in Group A but use the actual results in the other eight groups. It is worth noting that all teams play one match home and one away on the last two matchdays, which is not true for two subsequent matchdays chosen arbitrarily.

\begin{table}[ht!]
\centering
\caption{2018 FIFA World Cup qualification tournament, UEFA Group A---Results}
\label{Table1}

\begin{subtable}{\linewidth}
\centering
\caption{Match results of the first eight matchdays \\ \vspace{0.25cm}
\footnotesize{The positions are given according to the matches already played. \\
The home team is in the row, the away team (represented by its position) is in the column. \\
The dates are given for the matches to be played on the last two matchdays in 2017.}}
\label{Table1a}
\rowcolors{1}{}{gray!20}
    \begin{tabularx}{\linewidth}{cl CCC CCC} \toprule
    Position   & Team      & 1     & 2     & 3     & 4     & 5     & 6 \\ \hline
    1     & France  & ---     & 2-1   & 4-0   & 4-1   & 0-0   & 10 Oct \\
    2     & Sweden  & 2-1   & ---     & 1-1   & 3-0   & 7 Oct  & 4-0 \\
    3     & Netherlands  & 0-1   & 10 Oct  & ---     & 3-1   & 5-0   & 4-1 \\
    4     & Bulgaria  & 7 Oct  & 3-2   & 2-0   & ---     & 4-3   & 1-0 \\
    5     & Luxembourg & 1-3   & 0-1   & 1-3   & 10 Oct  & ---     & 1-0 \\
    6     & Belarus & 0-0   & 0-4   & 7 Oct  & 2-1   & 1-1   & --- \\ \bottomrule
    \end{tabularx}
\end{subtable}

\vspace{0.5cm}
\begin{subtable}{\linewidth}
\centering
\caption{Hypothetical match results of the last two matchdays \\ \vspace{0.25cm}
\footnotesize{The last row shows an alternative result, obtained if Bulgaria manipulates.}}
\label{Table1b}
\rowcolors{1}{}{gray!20}
    \begin{tabularx}{0.9\linewidth}{lLLc} \toprule
    Date  & Home team & Away team & Result \\ \hline \showrowcolors
    7 October 2017 & Sweden & Luxembourg & 0-4 \\
    7 October 2017 & Belarus & Netherlands  & 7-0 \\
    7 October 2017 & Bulgaria  & France & 8-0 \\ \hline
    10 October 2017 & France & Belarus & 1-0 \\
    10 October 2017 & Luxembourg & Bulgaria  & 0-1 \\
    10 October 2017 & Netherlands  & Sweden & 3-0 \\ \midrule \hiderowcolors
    10 October 2017* & Luxembourg* & Bulgaria*  & 1-1* \\\bottomrule
    \end{tabularx}
\end{subtable}
\end{table}

Table~\ref{Table1} shows a possible scenario in Group A. While some results of Table~\ref{Table1b} may be unreasonable, like Belarus defeating the Netherlands by 7-0, they are necessary to create the conditions for manipulation. Nevertheless, this set of match results had a positive probability after eight matchdays were over.

\begin{table}[ht!]
\centering
\caption{2018 FIFA World Cup qualification tournament, UEFA Group A---Standings \\ \vspace{0.25cm}
\footnotesize{Pos = Position; W = Won; D = Drawn; L = Lost; GF = Goals for; GA = Goals against; GD = Goal difference; Pts = Points. All teams have played $10$ matches. \\
The last row contains the second-placed team's adjusted results for the ranking of the runners-up (matches played against the sixth-placed team are discarded) according to \citet{FIFA2017a}.}}
\label{Table2}

\begin{subtable}{\linewidth}
\centering
\caption{Baseline standing with the runner-up results}
\label{Table2a}
\rowcolors{1}{}{gray!20}
    \begin{tabularx}{\linewidth}{Cl CCCC CC >{\bfseries}C} \toprule \showrowcolors
    Pos   & Team  & W     & D     & L     & GF    & GA    & GD    & Pts \\ \hline
    1     & France & 6     & 2     & 2     & 16    & 13    & 3     & 20 \\
    2     & Bulgaria & 6     & 0     & 4     & 22    & 17    & 5     & 18 \\
    3     & Sweden & 5     & 1     & 4     & 18    & 14    & 4     & 16 \\
    4     & Netherlands & 5     & 1     & 4     & 19    & 18    & 1     & 16 \\
    5     & Belarus & 2     & 2     & 6     & 11    & 17    & -6    & 8 \\
    6     & Luxembourg & 2     & 2     & 6     & 11    & 18    & -7    & 8 \\ \midrule \hiderowcolors
    2	  & Bulgaria & 4     & 0     & 4     & 17    & 14    & 3     & 12 \\ \bottomrule    
    \end{tabularx}
\end{subtable}

\vspace{0.5cm}
\begin{subtable}{\linewidth}
\centering
\caption{Alternative standing with the runner-up results if Bulgaria manipulates \\ \vspace{0.25cm}
\footnotesize{The changes are indicated with blue color and underlining.}}
\label{Table2b}
\rowcolors{1}{}{gray!20}
    \begin{tabularx}{\linewidth}{Cl CCCC CC >{\bfseries}C} \toprule \showrowcolors
    Pos   & Team  & W     & D     & L     & GF    & GA    & GD    & Pts \\ \hline
    1     & France & 6     & 2     & 2     & 16    & 13    & 3     & 20 \\
    2     & Bulgaria & \textcolor{blue}{\underline{5}}     & \textcolor{blue}{\underline{1}}     & 4     & 22    & \textcolor{blue}{\underline{18}}    & \textcolor{blue}{\underline{4}}     & \textcolor{blue}{\underline{16}} \\
    3     & Sweden & 5     & 1     & 4     & 18    & 14    & 4     & 16 \\
    4     & Netherlands & 5     & 1     & 4     & 19    & 18    & 1     & 16 \\
    \textcolor{blue}{\underline{5}}     & Luxembourg & 2     & \textcolor{blue}{\underline{3}}     & \textcolor{blue}{\underline{5}}     & \textcolor{blue}{\underline{12}}    & 18    & \textcolor{blue}{\underline{-6}}    & \textcolor{blue}{\underline{9}} \\ 
    \textcolor{blue}{\underline{6}}     & Belarus & 2     & 2     & 6     & 11    & 17    & -6    & 8 \\
\midrule \hiderowcolors
    2	  & Bulgaria & 4     & \textcolor{blue}{\underline{1}}     & \textcolor{blue}{\underline{3}}     & \textcolor{blue}{\underline{20}}    & \textcolor{blue}{\underline{16 }}  & \textcolor{blue}{\underline{4}}     & \textcolor{blue}{\underline{13}} \\ \bottomrule    
    \end{tabularx}
\end{subtable}
\end{table}

The final standing with the hypothetical results are shown in Table~\ref{Table2a}.

\begin{table}[ht!]
\centering
\caption[2018 FIFA World Cup qualification tournament (UEFA)---Ranking of second-placed teams]{2018 FIFA World Cup qualification (UEFA)---Ranking of second-placed teams \\ \vspace{0.25cm}
\footnotesize{Pos = Position; W = Won; D = Drawn; L = Lost; GF = Goals for; GA = Goals against; GD = Goal difference; Pts = Points. \\
Since matches played against the sixth-placed team in each group are discarded \citep{FIFA2017a}, all teams have played eight matches taken into account.}}
\label{Table3}

\begin{subtable}{\linewidth}
\centering
\caption{Baseline scenario}
\label{Table3a}
\rowcolors{1}{}{gray!20}
    \begin{tabularx}{\linewidth}{Clc CCCC CC >{\bfseries}C} \toprule \showrowcolors
    Pos   & Team  & Group & W     & D     & L     & GF    & GA    & GD    & Pts \\ \hline
    1     & Switzerland & B     & 7     & 0     & 1     & 18    & 6     & 12    & 21 \\
    2     & Italy & G     & 5     & 2     & 1     & 12    & 8     & 4     & 17 \\
    3     & Denmark & E     & 4     & 2     & 2     & 13    & 6     & 7     & 14 \\
    4     & Croatia & I     & 4     & 2     & 2     & 8     & 4     & 4     & 14 \\
    5     & Northern Ireland & C     & 4     & 1     & 3     & 10    & 6     & 4     & 13 \\
    6     & Greece & H     & 3     & 4     & 1     & 9     & 5     & 4     & 13 \\
    7     & Republic of Ireland & D     & 3     & 4     & 1     & 7     & 5     & 2     & 13 \\
    8     & Slovakia & F     & 4     & 0     & 4     & 11    & 6     & 5     & 12 \\ \bottomrule
    9     & Bulgaria  & A     & 4     & 0     & 4     & 17    & 14    & 3     & 12 \\ \toprule    
    \end{tabularx}
\end{subtable}

\vspace{0.5cm}
\begin{subtable}{\linewidth}
\centering
\caption{Alternative scenario if Bulgaria manipulates \\ \vspace{0.25cm}
\footnotesize{The changes are indicated with blue color and underlining.}}
\label{Table3b}
\rowcolors{1}{}{gray!20}
    \begin{tabularx}{\linewidth}{Clc CCCC CC >{\bfseries}C} \toprule \showrowcolors
    Pos   & Team  & Group & W     & D     & L     & GF    & GA    & GD    & Pts \\ \hline
    1     & Switzerland & B     & 7     & 0     & 1     & 18    & 6     & 12    & 21 \\
    2     & Italy & G     & 5     & 2     & 1     & 12    & 8     & 4     & 17 \\
    3     & Denmark & E     & 4     & 2     & 2     & 13    & 6     & 7     & 14 \\
    4     & Croatia & I     & 4     & 2     & 2     & 8     & 4     & 4     & 14 \\
    \textcolor{blue}{\underline{5}}     & Bulgaria & A     & 4     & \textcolor{blue}{\underline{1}}    & \textcolor{blue}{\underline{3}}    & \textcolor{blue}{\underline{20}}   & \textcolor{blue}{\underline{16}}   & \textcolor{blue}{\underline{4}}    & \textcolor{blue}{\underline{13}} \\
    \textcolor{blue}{\underline{6}}     & Northern Ireland & C     & 4     & 1     & 3     & 10    & 6     & 4     & 13 \\
    \textcolor{blue}{\underline{7}}     & Greece & H     & 3     & 4     & 1     & 9     & 5     & 4     & 13 \\
    \textcolor{blue}{\underline{8}}     & Republic of Ireland & D     & 3     & 4     & 1     & 7     & 5     & 2     & 13 \\ \bottomrule
    \textcolor{blue}{\underline{9}}     & Slovakia & F     & 4     & 0     & 4     & 11    & 6     & 5     & 12 \\ \toprule    
    \end{tabularx}
\end{subtable}
\end{table}

On the basis of final standing in Groups B--I and Table~\ref{Table2a}, the runners-up are ranked in Table~\ref{Table3a}. Since only the eight best second-placed teams advance to the play-offs, Bulgaria would be eliminated.

However, consider what happens if Bulgaria plays a draw of 1-1 against Luxembourg on the last matchday.
According to Table~\ref{Table2b}, while this change worsens Bulgaria's standing in the group, it remains the runner-up with $16$ points as both Bulgaria and Sweden would have the same goal difference ($+4$) with Bulgaria scoring more goals in all group matches ($22$ vs.\ $18$).\footnote{~Currently, UEFA prefers using head-to-head results instead of goal difference for breaking ties in the number of points. While this might have some influence, for example, on the probability that the relative ranking of two teams depends on a match in which neither team is involved \citep{Berker2014}, the choice of tie-breaking rules does not affect the essence of our example.}
On the other hand, Luxembourg would overtake Belarus thanks to its newly obtained draw because it would have more points ($9$ vs.\ $8$).

In the ranking of the second-placed teams, the matches played against the last-placed team are discarded. Consequently, Bulgaria would have $13$ points, placing it fifth among the runners-up as presented in Table~\ref{Table3b}. Thus Bulgaria would advance to the play-offs instead of Slovakia if it would concede a goal against Luxembourg.
\end{example}



Note that a successful manipulation has three requirements:
(1) the ranking criteria (number of points, goal difference, etc.) of a team should be better among the runners-up by exerting a lower effort in a match;
(2) it should preserve the position of the manipulating team in its group;
(3) it needs to result in a potential gain for this team with respect to qualification, for example, by advancing it to the play-offs instead of being eliminated.
Otherwise, it makes no sense to exert a lower effort.
A situation that occurred in the UEFA European Championship 1996 qualification tournament has satisfied only the first two conditions \citep{Csato2018b}.

\section{Theoretical background} \label{Sec3}

In the following, an abstract model is built for the home-and-away (double) round-robin group stage of a tournament designed similarly to the 2018 FIFA World Cup qualification tournament (UEFA). It begins with several definitions for the sake of accurateness. We will see in the next section that the seemingly long preparation for the main theorems does not affect the generality of the results. For example, it is not required that a win is awarded by three points, a draw is awarded by one point, and a loss is awarded by zero points.

\begin{definition} \label{Def301}
\emph{Round-robin group}:
The pair $(X,S)$ is a \emph{round-robin group} where
\begin{itemize}
\item
$X$ is a finite set of at least two teams;
\item
the \emph{ranking method} $S$ associates a strict order $S(v)$ on the set $X$ for any function $v: X \times X \to \left\{ \left( v_1; v_2 \right): v_1,v_2 \in \mathbb{N} \cup \{ 0 \} \right\} \cup \{ \text{---} \}$ such that $v(x,y) = \text{---}$ if and only if $x=y$.
\end{itemize}
\end{definition}

In a round-robin group, any team plays each other team in $X$ once at home and once at away (the home team and the away team is defined even if the match is played on a neutral field). Function $v$ describes game results with the number of goals scored by the home and the away team, respectively.

Let $(X,S)$ be a round-robin group, $x,y \in X$, $x \neq y$ be two teams, and $v$ be a set of results.
$x$ is ranked higher (lower) than $y$ if and only if $x$ is preferred to $y$ by $S(v)$, that is, $x \succ_{S(v)} y$ ($x \prec_{S(v)} y$).
Consider $v(x,y) = \left( v_1(x,y); v_2(x,y) \right)$, where $v_1(x,y)$ is the number of goals scored by the home team $x$ against the away team $y$, and $v_2(x,y)$ is the number of goals scored by the away team $y$ against the home team $x$.
It is said that team $x$ wins over team $y$ if $v_1(x,y) > v_2(x,y)$ (home) or $v_1(y,x) < v_2(y,x)$ (away), team $x$ loses to team $y$ if $v_1(x,y) < v_2(x,y)$ (home) or $v_1(y,x) > v_2(y,x)$ (away) and team $x$ draws against team $y$ if $v_1(x,y) = v_2(x,y)$ or $v_1(y,x) = v_2(y,x)$.

\begin{definition} \label{Def302}
\emph{Number of points}:
Let $(X,S)$ be a round-robin group, $x \in X$ be a team, $v$ be a set of results, and $\alpha > \beta > \gamma$ be three parameters.
Denote by $N_v^w(x)$ the number of wins, by $N_v^d(x)$ the number of draws, and by $N_v^\ell(x)$ the number of losses of team $x$, respectively.
The \emph{number of points} of team $x$ is $s_v(x) = \alpha N_v^w(x) + \beta N_v^d(x) + \gamma N_v^\ell(x)$.
\end{definition}

In other words, a win means $\alpha$ points, a draw means $\beta$ points and a loss means $\gamma$ points.

\begin{remark} \label{Rem31}
For the sake of simplicity, in the following, we will use the notation $(X,S)$ for a round-robin group $(X,S)$ with the parameters $\alpha$, $\beta$, $\gamma$.
Furthermore, it is assumed that $\alpha > \beta > \gamma$.
\end{remark}

\begin{definition} \label{Def303}
\emph{Monotonicity of the ranking in a round-robin group}:
Let $(X,S)$ be a round-robin group.
Its ranking method is called \emph{monotonic} if $s_v(x) > s_v(y)$ implies $x \succ_{S(v)} y$ for any teams $x,y \in X$ and for any set of results $v$.
\end{definition}

Monotonicity does not necessarily lead to a unique ranking, tie-breaking rules can be arbitrary in the model.

\begin{definition} \label{Def304}
\emph{Group-based qualifier}:
A \emph{group-based qualifier} $\mathcal{T}$ consists of $k$ round-robin groups $\left( X^i, S^i \right)$, $1 \leq i \leq k$ such that $X^i \cap X^j = \emptyset$ for any $i \neq j$, $1 \leq i,j \leq k$.
\end{definition}

To cover the 2018 FIFA World Cup qualification tournament (UEFA), it is allowed to compare teams from different groups in the \emph{repechage group}.

\begin{definition} \label{Def305}
\emph{Repechage function}:
Let $\mathcal{T}$ be a group-based qualifier.
For any set of group results $V = \left\{ v^1, v^2, \dots , v^k \right\}$, a \emph{repechage function} $\mathcal{G}$ associates
\begin{itemize}
\item
a set of teams $\mathcal{G}_1 (V) \subseteq \cup_{i=1}^k X^i$ composing the repechage group;
\item
a set of opponents $\emptyset \neq \mathcal{G}_2 (V,x) \subseteq X^i \setminus \{ x \}$ for each team in the repechage group $x \in \mathcal{G}_1 (V)$; and
\item
a \emph{repechage ranking} $Q(V)$, which is a strict order $Q(V)$ on the set $\mathcal{G}_1(V)$.
\end{itemize}
\end{definition}

\begin{definition} \label{Def306}
\emph{Impartiality of a repechage function}:
Let $\mathcal{T}$ be a group-based qualifier.
Repechage function $\mathcal{G}$ is \emph{impartial} if for all $1 \leq i \leq k$:
\begin{enumerate}[label=\alph*)]
\item \label{Impartiality_con1}
$\left| X^i \cap \mathcal{G}_1 (V) \right| = c^i$, the number of teams from the round-robin group $(X^i, S^i)$ in the repechage group, does not depend on the set of group results $V = \left\{ v^1, v^2, \dots , v^k \right\}$;
\item \label{Impartiality_con2}
$x,y \in X^i \cap \mathcal{G}_1 (V)$ implies $y \in \mathcal{G}_2 (V,x)$ and $x \in \mathcal{G}_2 (V,y)$;
\item \label{Impartiality_con3}
if $x,y \in X^i \cap \mathcal{G}_1 (V)$, then $z \in \mathcal{G}_2 (V,x)$, $z \neq y$ implies $z \in \mathcal{G}_2 (V,y)$;
and
\item \label{Impartiality_con4}
$x,y \in \mathcal{G}_1 (V)$ implies $\left| \mathcal{G}_2 (V,x) \right| = \left| \mathcal{G}_2 (V,y) \right|$.
\end{enumerate}
\end{definition}

According to the impartiality of a repechage function:
\ref{Impartiality_con1} the number of teams transferred to the repechage group from a given group is fixed;
\ref{Impartiality_con2} if two teams in the repechage group have played against each other, then these matches are considered in the repechage group; 
\ref{Impartiality_con3} if two teams are transferred to the repechage group from the same group, then their matches played against any third team should be uniformly considered or discarded in the repechage group; 
and
\ref{Impartiality_con4} the number of matches taken into account for the repechage group is the same for all teams in the repechage group.
The last condition ensures that the number of points is a plausible measure of performance in the repechage group.

\begin{definition} \label{Def307}
\emph{Number of points in the repechage group}:
Let $\mathcal{T}$ be a group-based qualifier and $\mathcal{G}$ be a repechage function.
The \emph{number of points in the repechage group} of team $x \in \mathcal{G}_1(V)$ is:
\begin{eqnarray*}
s_{\mathcal{G}(V)}^{k+1}(x) & = & \alpha \left( \left| \left\{ y \in \mathcal{G}_2 (V,x): v_1^i(x,y) > v_2^i(x,y) \quad \text{or} \quad v_1^i(y,x) < v_2^i(y,x) \right\} \right| \right) + \\
& & + \beta \left( \left| \left\{ y \in \mathcal{G}_2 (V,x): v_1^i(x,y) = v_2^i(x,y) \quad \text{or} \quad v_1^i(y,x) = v_2^i(y,x) \right\} \right| \right) + \\
& & + \gamma \left( \left| \left\{ y \in \mathcal{G}_2 (V,x): v_1^i(x,y) < v_2^i(x,y) \quad \text{or} \quad v_1^i(y,x) > v_2^i(y,x) \right\} \right| \right).
\end{eqnarray*}
\end{definition}

Thus the number of points in the repechage group is calculated on the basis of the matches considered there.

\begin{definition} \label{Def308}
\emph{Monotonicity of a repechage function}:
Let $\mathcal{T}$ be a group-based qualifier and $\mathcal{G}$ be a repechage function.
Repechage function $\mathcal{G}$ is said to be \emph{monotonic} if $s_{\mathcal{G}(V)}^{k+1}(x) > s_{\mathcal{G}(V)}^{k+1}(y)$ implies $x \succ_{Q(V)} y$ for any teams $x,y \in \mathcal{G}_1(V)$ and any set of group results $V = \left\{ v^1, v^2, \dots , v^k \right\}$.
\end{definition}

\begin{definition} \label{Def309}
\emph{Allocation rule}:
Let $\mathcal{T}$ be a group-based qualifier and $\mathcal{G}$ be a repechage function.
For any set of group results $V = \left\{ v^1, v^2, \dots , v^k \right\}$, an \emph{allocation rule} $\mathcal{G}$ associates a value from the set $\left\{ 0; 1; 2 \right\}$ for all teams $x \in \cup_{i=1}^k X^i$.
\end{definition}

Consider a group-based qualifier $\mathcal{T}$, a repechage function $\mathcal{G}$, an allocation rule $\mathcal{R}$, and a set of group results $V = \left\{ v^1, v^2, \dots , v^k \right\}$.
Team $x \in \cup_{i=1}^k X^i$ is said to be: (a) directly qualified if $\mathcal{R}(V,x) = 2$; (b) advanced to the next round (with a chance to qualify) if $\mathcal{R}(V,x) = 1$; and (c) eliminated if $\mathcal{R}(V,x) = 0$.

\begin{definition} \label{Def310}
\emph{Groups-based qualification system}:
The triple $(\mathcal{T}, \mathcal{G}, \mathcal{R})$ of a group-based qualifier $\mathcal{T}$, a repechage function $\mathcal{G}$, and an allocation rule $\mathcal{R}$ is a \emph{group-based qualification system}.
\end{definition}

The outcome of any group-based qualification system $(\mathcal{T}, \mathcal{G}, \mathcal{R})$ is given by the allocation rule $\mathcal{R}$, which implies the following.

\begin{remark} \label{Rem32}
Let $(\mathcal{T}, \mathcal{G}, \mathcal{R})$ be a group-based qualification system.
If there is no difference in the allocation of teams in the repechage group, that is, $\mathcal{R}(V,x) = \mathcal{R}(V,y)$ for all teams $x,y \in \mathcal{G}_1(V)$ under any set of group results $V = \left\{ v^1, v^2, \dots , v^k \right\}$, then the group-based qualification system can be equivalently described without a repechage group, thus $\mathcal{G}_1(V) = \emptyset$ can be assumed without loss of generality.
\end{remark}

\begin{definition} \label{Def311}
\emph{Monotonicity of a group-based qualification system}:
Let $(\mathcal{T}, \mathcal{G}, \mathcal{R})$ be a group-based qualification system.
It is called \emph{monotonic} if under any set of group results $V = \left\{ v^1, v^2, \dots , v^k \right\}$:
\begin{itemize}
\item
there exists a common monotonic ranking $S = S^i$ for each round-robin group $\left( X^i, S^i \right)$ such that $x \succ_{S(v^i)} y$ implies $\mathcal{R}(V,x) \geq \mathcal{R}(V,y)$ for all $x,y \in X^i$, $1 \leq i \leq k$;
\item
the repechage function $\mathcal{G}$ impartial and monotonic such that $x \succ_{Q(V)} y$ implies $\mathcal{R}(V,x) \geq \mathcal{R}(V,y)$ for all $x,y \in \mathcal{G}_1(V)$.
\end{itemize}
\end{definition}

The idea behind a monotonic group-based qualification system is straightforward.
Because of the application of a monotonic group ranking, teams have no incentive to exert a lower effort in any match as they cannot achieve a higher position in their group by deliberately playing worse.
Impartiality of the repechage function provides the comparability of the results taken into account for the repechage group, where the ranking ensures incentive compatibility, too.

\begin{proposition} \label{Prop31}
2018 FIFA World Cup qualification tournament (UEFA), discussed in Section~\ref{Sec2}, fits into the model presented above: it is a monotonic group-based qualification system.
\end{proposition}

\begin{proof}
There are $k=9$ groups and all criteria of Definition~\ref{Def311} hold:
\begin{itemize}
\item
the common group ranking $S$ is monotonic because the number of points is the first tie-breaker in the groups (Definition~\ref{Def303}), and:
\begin{itemize}[label=$\diamond$]
\item
the first-placed team in each group qualifies: $\mathcal{R}(V,x) = 2$ for each $x \in X^i$ if and only if $\left| \left\{ y \in X^i: y \succ_{S(v^i)} x \right\} \right| = 0$;
\item
the third-, fourth-, fifth- and sixth-placed teams in each group are eliminated: $\mathcal{R}(V,x) = 0$ for each $x \in X^i$ if $\left| \left\{ y \in X^i: y \succ_{S(v^i)} x \right\} \right| \geq 2$;
\end{itemize}
\item
the repechage function $\mathcal{G}$ is impartial (Definition~\ref{Def306}):
\begin{itemize}[label=$\diamond$]
\item
the repechage group consists of the runners-up: $x \in X^i \cap \mathcal{G}_1(V)$, $1 \leq i \leq 9$ if and only if $\left| \left\{ y \in X^i: y \succ_{S(v^i)} x \right\} \right| = 1$, so $\left| X^i \cap \mathcal{G}_1(V) \right| = 1$;
\item
the eight matches played against the rest of the first five teams of the group are considered in the repechage group: for each $x \in X^i \cap \mathcal{G}_1(V)$, $1 \leq i \leq 9$, $z \in \mathcal{G}_2(V,x)$ if and only if $z \in X^i$ and $\left| \left\{ y \in X^i: y \succ_{S(v^i)} z \right\} \right| \leq 4$;
\end{itemize}
\item
the repechage function $\mathcal{G}$ is monotonic because the number of points is the first tie-breaker in the repechage group (Definition~\ref{Def308}), and:
\begin{itemize}[label=$\diamond$]
\item
the eight best second-placed teams advance to the next round: $\mathcal{R}(V,x) = 1$ if $x \in \mathcal{G}_1(V)$ and $\left| \left\{ y \in \mathcal{G}_1(V): y \succ_{Q(V)} x \right\} \right| \leq 7$;
and
\item
the lowest-ranked second-placed team is eliminated: $\mathcal{R}(V,x) = 0$ if $x \in \mathcal{G}_1(V)$ and $y \succ_{Q(V)} x$ for all $y \in \mathcal{G}_1(V) \setminus \{ x \}$.
\end{itemize}
\end{itemize}
\end{proof}

Now we turn to the issue of incentive compatibility.

\begin{definition} \label{Def312}
\emph{Manipulation}:
Let $(\mathcal{T}, \mathcal{G}, \mathcal{R})$ be a group-based qualification system.
A team $x \in X^i$, $1 \leq i \leq k$ can \emph{manipulate} the qualification system if there exist two sets of group results $V = \left\{ v^1, v^2, \dots , v^i, \dots , v^k \right\}$ and $\bar{V} = \left\{ v^{1}, v^2, \dots , \bar{v}^i, \dots , v^k \right\}$ such that $v_2^i(x,y) \leq \bar{v}_2^i(x,y)$ and $v_1^i(y,x) \leq \bar{v}_1^i(y,x)$ for all $y \in X^i$, furthermore, $\mathcal{R}(V,x) < \mathcal{R}(\bar{V},x)$.
\end{definition}

To summarise, manipulation means that team $x$ improves its status with respect to qualification by letting one of its opponents to score more goals.

\begin{definition} \label{Def313}
\emph{Strategy-proofness of a group-based qualification system}:
A group-based qualification system $(\mathcal{T}, \mathcal{G}, \mathcal{R})$ is called \emph{strategy-proof} if there exists no set of group results $V = \left\{ v^1, v^2, \dots ,v^k \right\}$ under which a team $x \in \cup_{i=1}^k X^i$ can manipulate.
\end{definition}

Another notion is needed to formulate the theoretical results.

\begin{definition} \label{Def314}
\emph{Permutation of results in a group}:
Let $\mathcal{T}$ be a group-based qualifier and $V = \left\{ v^1, v^2, \dots , v^i, \dots , v^k \right\}$ be a set of group results.
Consider a round-robin group $\left( X^i, S^i \right)$, $1 \leq i \leq k$, and a permutation $\sigma: X^i \to X^i$.
\emph{Permutation of results in a group} leads to the set of group results $\left\{ v^1, v^2, \dots , \sigma(v)^i, \dots , v^k \right\}$ such that $\sigma(v)^i(x,y) = v^i \left( \sigma^{-1}(x),\sigma^{-1}(y) \right)$ for all teams $x,y \in X^i$.
\end{definition}

Permutation of group results is illustrated by the following example.

\begin{example} \label{Examp31}
Consider the 2018 FIFA World Cup qualification (UEFA), which is a (monotonic) group-based qualification system according to Proposition~\ref{Prop31}, and the set of group results $V$ given by the actual results in Groups B--I and the hypothetical group results in Group A according to Table~\ref{Table1}.
Let $\sigma$ be the permutation such that $\sigma (\text{France}) = \text{Sweden}$, $\sigma (\text{Sweden}) = \text{Netherlands}$, and $\sigma (\text{Netherlands}) = \text{France}$, while $\sigma(y) = y$ for all teams $y \notin \{ \text{France, Netherlands, Sweden} \}$.
Then France would have the same results under $\sigma(V)$ as Sweden has under $V$, Sweden would have the same results under $\sigma(V)$ as Netherlands has under $V$, and Netherlands would have the same results under $\sigma(V)$ as France has under $V$. Consequently, Netherlands would qualify under $\sigma(V)$ as the group winner, France would have a goal difference of $+4$, and Sweden would have a goal difference of $+1$ with $16$ points (see Table~\ref{Table2a}).
\end{example}

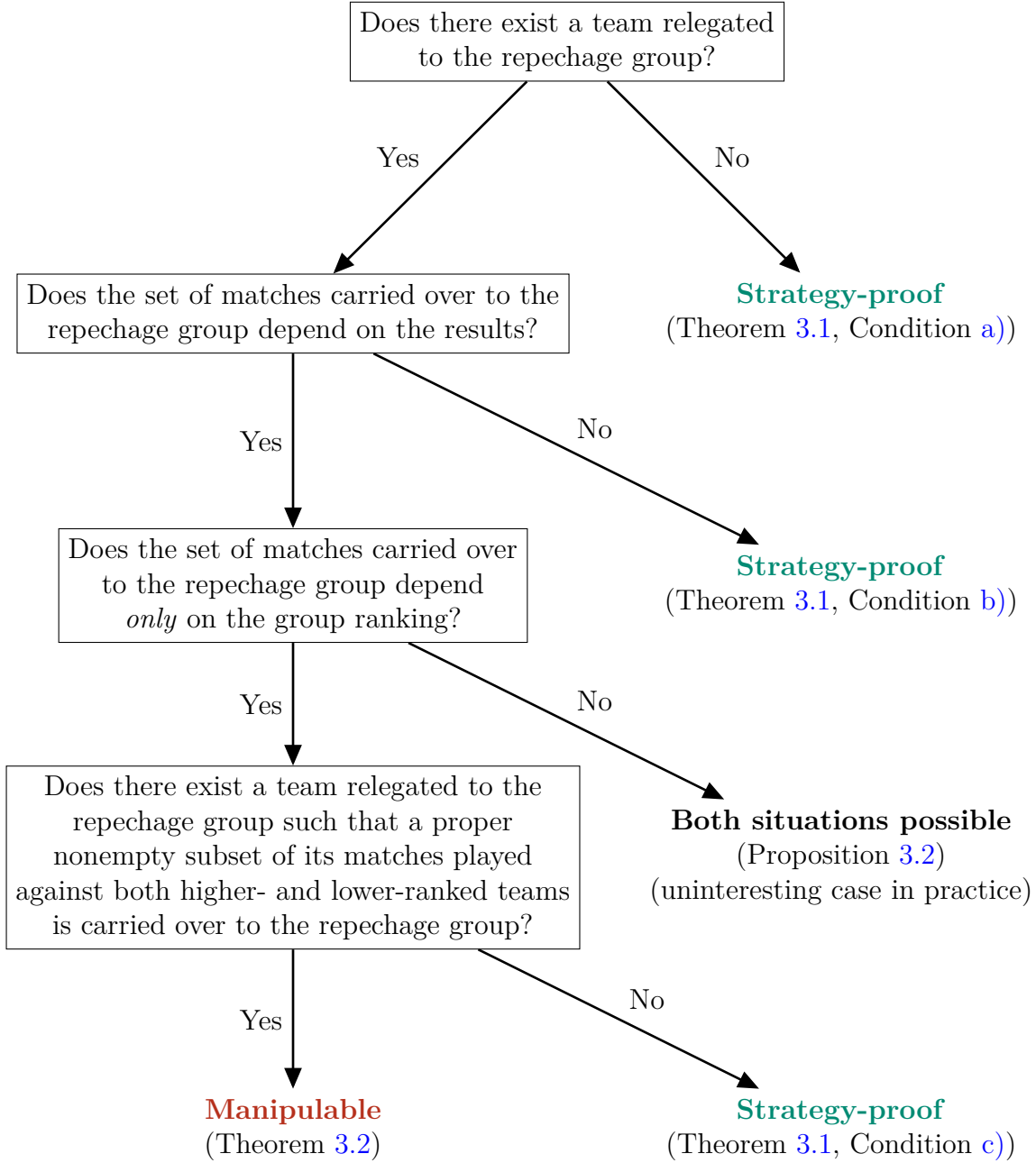
\begin{figure}[ht!]
\centering
\begin{tikzpicture}[scale=1,auto=center, transform shape, >=triangle 45]
\tikzstyle{every node}=[draw,align=center];
  \node (C1) at (0,16) {Does there exist a team relegated \\ to the repechage group?};
  \node (C2) at (-4,12) {Does the set of matches carried over to the \\ repechage group depend on the results?};
  \node (C4) at (-4,8) {Does the set of matches carried over \\ to the repechage group depend \\ \emph{only} on the group ranking?};
  \node (C6) at (-4,4) {Does there exist a team relegated to the \\ repechage group such that a proper \\ nonempty subset of its matches played \\ against both higher- and lower-ranked teams \\ is carried over to the repechage group?};
  
\tikzstyle{every node}=[align=center];
  \node[align=center] at (0,18) {\textbf{The strategy-proofness of a monotonic group-based qualification system} \\ (The conditions should be checked separately for each group)};
  \node (C3) at (4,12)   {\textcolor{PineGreen}{\textbf{Strategy-proof}} \\ (Theorem~\ref{Theo31}, Condition~\ref{Con1})};
  \node (C5) at (4,8)   {\textcolor{PineGreen}{\textbf{Strategy-proof}} \\ (Theorem~\ref{Theo31}, Condition~\ref{Con2})};
  \node (C7) at (4,4)   {\textbf{Both situations possible} \\ (Proposition~\ref{Prop32}) \\ (uninteresting case in practice)};
  \node (C8) at (-4,0) {\textcolor{BrickRed}{\textbf{Manipulable}} \\ (Theorem~\ref{Theo32})};
  \node (C9) at (4,0)  {\textcolor{PineGreen}{\textbf{Strategy-proof}} \\ (Theorem~\ref{Theo31}, Condition~\ref{Con3})};

\tikzstyle{every node}=[align=center];  
  \draw [->,line width=1pt] (C1) -- (C2)  node [midway,above left] {Yes};
  \draw [->,line width=1pt] (C1) -- (C3)  node [midway,above right] {No};
  \draw [->,line width=1pt] (C2) -- (C4)  node [midway,left] {Yes};
  \draw [->,line width=1pt] (C2) -- (C5)  node [midway,above right] {No};
  \draw [->,line width=1pt] (C4) -- (C6)  node [midway,left] {Yes};
  \draw [->,line width=1pt] (C4) -- (C7)  node [midway,above right] {No};
  \draw [->,line width=1pt] (C6) -- (C8)  node [midway,left] {Yes};
  \draw [->,line width=1pt] (C6) -- (C9)  node [midway,above right] {No};
\end{tikzpicture}
\caption{A summary of the main result}
\label{Fig1}
\end{figure}


Our main contribution concerns the strategy-proofness of a monotonic group-based qualification system, summarised in Figure~\ref{Fig1}.

\begin{theorem} \label{Theo31}
Let $(\mathcal{T}, \mathcal{G}, \mathcal{R})$ be a monotonic group-based qualification system such that at least one of the following conditions hold for all round-robin groups $\left( X^i, S \right)$, $1 \leq i \leq k$:
\begin{enumerate}[label=\alph*)]
\item \label{Con1}
no team is transferred to the repechage group:
$\left| X^i \cap \mathcal{G}_1(V) \right| = 0$;
\item \label{Con2}
for each team $x \in X^i \cap \mathcal{G}_1(V)$, 
its set of matches considered in the repechage group is independent of the set of group results $V$:
$\mathcal{G}_2(V,x) = \mathcal{G}_2(\sigma(V),x)$ for any permutation of results in this group $\sigma$ that satisfies $x \in \left( X^i \cap \mathcal{G}_1(V) \cap \mathcal{G}_1(\sigma(V)) \right)$;
\item \label{Con3}
for each team $x \in X^i \cap \mathcal{G}_1(V)$, 
its set of matches considered in the repechage group depends only on the group ranking and either all or none of its matches played against both higher- and lower-ranked teams are considered in the repechage group under any set of group results $V$:
$\mathcal{G}_2(\hat{V},x) = \mathcal{G}_2(\tilde{V},x)$ if $y \succ_{S(\hat{v}^i)} z$ implies $y \succ_{S(\tilde{v}^i)} z$ for any sets of group results $\hat{V}, \tilde{V}$ and for all teams $y,z \in X^i$, furthermore, $x \prec_{S(v^i)} y$ implies $y \in \mathcal{G}_2(V,x)$ or $y \notin \mathcal{G}_2(V,x)$, and $x \succ_{S(v^i)} y$ implies $y \in \mathcal{G}_2(V,x)$ or $y \notin \mathcal{G}_2(V,x)$ for all teams $y \in X^i$. 
\end{enumerate}
Then this group-based qualification system is strategy-proof.
\end{theorem}

See Remark~\ref{Rem32} for the first condition.
According to the second requirement, if a permutation of group results does not affect the relegation of a given team to the repechage group, then its set of matches carried over to the repechage group should remain the same.
The last condition of Theorem~\ref{Theo31} means that all or none of matches played against higher- and lower-ranked teams should be considered in the repechage group.

\begin{proof}
\ref{Con1}
The monotonicity of the group ranking $S$ provides incentive compatibility.

\ref{Con2}
Since $\mathcal{G}_2(V,x) = \mathcal{G}_2(\bar{V},x)$ for team $x \in X^i \cap \mathcal{G}_1(V)$ under any sets of group results $V,\bar{V}$, incentive compatibility is provided by the monotonicity of the group ranking $S$ and the repechage function $\mathcal{G}$.

\ref{Con3}
In this case, the team $x \in X^i \cap \mathcal{G}_1(V)$ cannot change the set of its matches to be discarded in the repechage group, therefore incentive compatibility is provided by the monotonicity of the group ranking $S$ and repechage function $\mathcal{G}$.
\end{proof}

The message of Theorem~\ref{Theo31} in practice will be discussed later.

\begin{theorem} \label{Theo32}
Let $(\mathcal{T}, \mathcal{G}, \mathcal{R})$ be a monotonic group-based qualification system such that the following conditions hold for at least one round-robin group $\left( X^i, S \right)$, $1 \leq i \leq k$:
\begin{enumerate}[label=\alph*)]
\item
a team is transferred to the repechage group:
$\left| X^i \cap \mathcal{G}_1(V) \right| > 0$;\footnote{~According to Remark~\ref{Rem32}, this implies that there is a difference in the allocation of teams from the repechage group: if $x \in \mathcal{G}_1(V)$, then there exists at least one team $y \in \mathcal{G}_1(V)$ such that $\mathcal{R}(V,x) \neq \mathcal{R}(V,y)$.}
\item
for each team $x \in X^i \cap \mathcal{G}_1(V)$,
its set of matches considered in the repechage group depends only on the group ranking:
$\mathcal{G}_2(\hat{V},x) = \mathcal{G}_2(\tilde{V},x)$ if $y \succ_{S(\hat{v}^i)} z$ implies $y \succ_{S(\tilde{v}^i)} z$ for any sets of group results $\hat{V}, \tilde{V}$ and all teams $y,z \in X^i$;
and
\item
there exists a team $x \in X^i \cap \mathcal{G}_1(V)$ such that a proper nonempty subset of its matches played against higher- or lower-ranked teams are considered in the repechage group under any set of group results $V$:
\[
\emptyset \neq \left( \mathcal{G}_2(V,x) \cap \left\{ y \in X^i: y \succ_{S(v^i)} x \right\} \right) \subset \left\{ y \in X^i: y \succ_{S(v^i)} x \right\}, \text{ or}
\]
\[
\emptyset \neq \left( \mathcal{G}_2(V,x) \cap \left\{ y \in X^i: x \succ_{S(v^i)} y \right\} \right) \subset \left\{ y \in X^i: x \succ_{S(v^i)} y \right\}.
\]
\end{enumerate}
Then this group-based qualification system $(\mathcal{T}, \mathcal{G}, \mathcal{R})$ violates strategy-proofness.
\end{theorem}

The requirements of Theorem~\ref{Theo32} are outlined in Figure~\ref{Fig1}, too, as this result runs parallel with Theorem~\ref{Theo31}.

No formal proof is given for Theorem~\ref{Theo32}, but an example with a group of four teams is provided, where the second-placed team can manipulate.


\begin{table}[ht!]
\centering
\caption[The group of Example~\ref{Examp32}]{The group of Example~\ref{Examp32} \\ \vspace{0.25cm}
\footnotesize{The runner-up is transferred to the repechage group. \\
The home team is in the row, the result of the match is given from its point of view.}}
\label{Table4}

\begin{subtable}{0.44\textwidth}
\centering
\caption{A situation susceptible to manipulation \\ if the matches played against the third-placed team count (Case~\ref{Case2}), or the matches played against the fourth-placed team are discarded (Case~\ref{Case6}) in the repechage group}
\label{Table4a}
\rowcolors{1}{}{gray!20}
    \begin{tabularx}{\textwidth}{C cccc} \toprule
    Team  & A     & B     & C     & D \\ \hline
    A     & ---   & draw  & win   & win \\
    B     & win   & ---   & win   & draw \\
    C     & loss  & loss  & ---   & win \\
    D     & loss  & win   & draw  & --- \\ \bottomrule
    \end{tabularx}
\end{subtable}
\hspace{0.1\textwidth}
\begin{subtable}{0.44\textwidth}
\centering
\caption{A situation susceptible to manipulation \\ if the matches played against the fourth-placed team count (Case~\ref{Case3}), or the matches played against the third-placed team are discarded (Case~\ref{Case5}) in the repechage group}
\label{Table4b}
\rowcolors{1}{}{gray!20}
    \begin{tabularx}{\textwidth}{C cccc} \toprule
    Team  & A     & B     & C     & D \\ \hline
    A     & ---   & win   & win   & win \\
    B     & win   & ---   & win   & win \\
    C     & draw  & loss  & ---   & win \\
    D     & loss  & win   & draw  & --- \\ \bottomrule
    \end{tabularx}
\end{subtable}
\end{table}

\begin{example} \label{Examp32}
Let $X^i = \{ A,B,C,D \}$.
Consider the monotonic group-based qualification system $(\mathcal{T}, \mathcal{G}, \mathcal{R})$ where $x \in X^i \cap \mathcal{G}_1(V)$ if $\left| \left\{ y \in X^i: y \succ_{S(v^i)} x \right\} \right| = 1$, therefore the runner-up is transferred to the repechage group.
Since there is a difference in the allocation of teams in the repechage group (the first requirement of Theorem~\ref{Theo32} holds), the results in other groups can be chosen such that the second-placed team can manipulate if its number of points considered in the repechage group increases.

Let $x \in X^i \cap \mathcal{G}_1(V)$.
The following cases are possible as the second condition of Theorem~\ref{Theo32} holds:
\begin{enumerate}[label=\emph{\Roman*}., ref=\emph{\Roman*}]
\item \label{Case1}
$\mathcal{G}_2(V,x) = \left\{ y \in X^i: \left| \left\{ z \in X^i: z \succ_{S(v^i)} y \right\} \right| = 0 \right\}$, that is, only the matches played against the top team count in the repechage group. \\
Condition~\ref{Con3} of Theorem~\ref{Theo31} provides strategy-proofness. \\
Note that
\[
\left( \mathcal{G}_2(V,x) \cap \left\{ y \in X^i: y \succ_{S(v^i)} x \right\} \right) = \left\{ y \in X^i: y \succ_{S(v^i)} x \right\} \text{, and}
\]
\[
\left( \mathcal{G}_2(V,x) \cap \left\{ y \in X^i: x \succ_{S(v^i)} y \right\} \right) = \emptyset,
\]
thus the third requirement of Theorem~\ref{Theo32} does not hold.
\item \label{Case2}
$\mathcal{G}_2(V,x) = \left\{ y \in X^i: \left| \left\{ z \in X^i: z \succ_{S(v^i)} y \right\} \right| = 2 \right\}$, that is, only the matches played against the third-placed team count in the repechage group. \\
Consider the results $v^i$ given in Table~\ref{Table4a}. The number of points are $s_{v^i}(A) = 4 \alpha + \beta + \gamma$, $s_{v^i}(B) = 3 \alpha + 2 \beta + \gamma$, $s_{v^i}(C) = \alpha + \beta + 4 \gamma$, and $s_{v^i}(D) = \alpha + 2 \beta + 3 \gamma$, hence $A \succ_{S(v^i)} B \succ_{S(v^i)} D \succ_{S(v^i)} C$ due to the monotonicity of the group ranking $S$. Furthermore, team $B$ has $\beta + \gamma$ points in the repechage group. \\
However, if $\bar{v}^i = v^i$ except for $\bar{v}^i(B,C) = \{ \text{loss} \}$ instead of $v^i(B,C) = \{ \text{win} \}$, then $s_{\bar{v}^i}(A) = s_{v^i}(A) = 4 \alpha + \beta + \gamma$, $s_{\bar{v}^i}(B) = 2 \alpha + 2 \beta + 2 \gamma$, $s_{\bar{v}^i}(C) = 2 \alpha + \beta + 3 \gamma$, and $s_{\bar{v}^i}(D) = s_{v^i}(D) = \alpha + 2 \beta + 3 \gamma$, hence $A \succ_{S(\bar{v}^i)} B \succ_{S(\bar{v}^i)} C \succ_{S(\bar{v}^i)} D$ due to the monotonicity of the group ranking $S$. Now team $B$ has $\alpha  + \gamma > \beta + \gamma$ points in the repechage group, therefore it can manipulate. \\
Note that
\[
\emptyset \neq \left( \mathcal{G}_2(V,x) \cap \left\{ y \in X^i: x \succ_{S(v^i)} y \right\} \right) \subset \left\{ y \in X^i: x \succ_{S(v^i)} y \right\},
\]
thus the third requirement of Theorem~\ref{Theo32} holds.
\item \label{Case3}
$\mathcal{G}_2(V,x) = \left\{ y \in X^i: \left| \left\{ z \in X^i: z \succ_{S(v^i)} y \right\} \right| = 3 \right\}$, that is, only the matches played against the fourth-placed team count in the repechage group. \\
Consider the results $v^i$ given in Table~\ref{Table4b}. The number of points are $s_{v^i}(A) = 4 \alpha + \beta + \gamma$, $s_{v^i}(B) = 4 \alpha + 2 \gamma$, $s_{v^i}(C) = \alpha + 2 \beta + 3 \gamma$, and $s_{v^i}(D) = \alpha + \beta + 4 \gamma$, hence $A \succ_{S(v^i)} B \succ_{S(v^i)} C \succ_{S(v^i)} D$ due to the monotonicity of the group ranking $S$. Furthermore, team $B$ has $\alpha + \gamma$ points in the repechage group. \\
However, if $\bar{v}^i = v^i$ except for $\bar{v}^i(B,D) = \{ \text{loss} \}$ instead of $v^i(B,D) = \{ \text{win} \}$, then $s_{\bar{v}^i}(A) = s_{v^i}(A) = 4 \alpha + \beta + \gamma$, $s_{\bar{v}^i}(B) = 3 \alpha + 3 \gamma$, $s_{\bar{v}^i}(C) = s_{v^i}(C) = \alpha + 2 \beta + 3 \gamma$, and $s_{\bar{v}^i}(D) = 2 \alpha + \beta + 3 \gamma$, hence $A \succ_{S(\bar{v}^i)} B \succ_{S(\bar{v}^i)} D \succ_{S(\bar{v}^i)} C$ due to the monotonicity of the group ranking $S$. Now team $B$ has $2 \alpha > \alpha + \gamma$ points in the repechage group, therefore it can manipulate. \\
Note that
\[
\emptyset \neq \left( \mathcal{G}_2(V,x) \cap \left\{ y \in X^i: x \succ_{S(v^i)} y \right\} \right) \subset \left\{ y \in X^i: x \succ_{S(v^i)} y \right\},
\]
thus the third requirement of Theorem~\ref{Theo32} holds.
\item \label{Case4}
$\mathcal{G}_2(V,x) = X^i \setminus \{ x \} \setminus \left\{ y \in X^i: \left| \left\{ z \in X^i: z \succ_{S(v^i)} y \right\} \right| = 0 \right\}$, that is, the matches played against the top team are discarded in the repechage group. \\
Condition~\ref{Con3} of Theorem~\ref{Theo31} provides strategy-proofness. \\
Note that
\[
\left( \mathcal{G}_2(V,x) \cap \left\{ y \in X^i: y \succ_{S(v^i)} x \right\} \right) = \emptyset \text{, and}
\]
\[
\left( \mathcal{G}_2(V,x) \cap \left\{ y \in X^i: x \succ_{S(v^i)} y \right\} \right) = \left\{ y \in X^i: x \succ_{S(v^i)} y \right\},
\]
thus the third requirement of Theorem~\ref{Theo32} does not hold.
\item \label{Case5}
$\mathcal{G}_2(V,x) = X^i \setminus \{ x \} \setminus \left\{ y \in X^i: \left| \left\{ z \in X^i: z \succ_{S(v^i)} y \right\} \right| = 2 \right\}$, that is, the matches played against the third-placed team are discarded in the repechage group. \\
Consider the results $v^i$ given in Table~\ref{Table4b} and the alternative results $\bar{v}^i$ from Case~\ref{Case3}.
The analysis of Case~\ref{Case3} remains valid, and team $B$ has $2 \alpha + 2 \gamma$ points under $v^i$ and $3 \alpha + \gamma$ points under $\bar{v}^i$ in the repechage group, respectively, therefore it can manipulate. \\
Note that
\[
\emptyset \neq \left( \mathcal{G}_2(V,x) \cap \left\{ y \in X^i: x \succ_{S(v^i)} y \right\} \right) \subset \left\{ y \in X^i: x \succ_{S(v^i)} y \right\},
\]
thus the third requirement of Theorem~\ref{Theo32} holds.
\item \label{Case6}
$\mathcal{G}_2(V,x) = X^i \setminus \{ x \} \setminus \left\{ y \in X^i: \left| \left\{ z \in X^i: z \succ_{S(v^i)} y \right\} \right| = 3 \right\}$, that is, the matches played against the fourth-placed team are discarded in the repechage group. \\
Consider the results $v^i$ given in Table~\ref{Table4a} and the alternative results $\bar{v}^i$ from Case~\ref{Case2}.
The analysis of Case~\ref{Case2} remains valid, and team $B$ has $\alpha + 2 \beta + \gamma$ points under $v^i$ and $2 \alpha + \beta + \gamma$ points under $\bar{v}^i$ in the repechage group, respectively, therefore it can manipulate. \\
Note that
\[
\emptyset \neq \left( \mathcal{G}_2(V,x) \cap \left\{ y \in X^i: x \succ_{S(v^i)} y \right\} \right) \subset \left\{ y \in X^i: x \succ_{S(v^i)} y \right\},
\]
thus the third requirement of Theorem~\ref{Theo32} holds.
\item \label{Case7}
$\mathcal{G}_2(V,x) = X^i \setminus \{ x \}$, that is, no matches are discarded in the repechage group. \\
Condition~\ref{Con3} of Theorem~\ref{Theo31} provides strategy-proofness. \\
Note that
\[
\left( \mathcal{G}_2(V,x) \cap \left\{ y \in X^i \setminus \mathcal{G}_1(V): y \succ_{S(v^i)} x \right\} \right) = \left\{ y \in X^i \setminus \mathcal{G}_1(V): y \succ_{S(v^i)} x \right\} \text{, and}
\]
\[
\left( \mathcal{G}_2(V,x) \cap \left\{ y \in X^i \setminus \mathcal{G}_1(V): x \succ_{S(v^i)} y \right\} \right) = \left\{ y \in X^i \setminus \mathcal{G}_1(V): x \succ_{S(v^i)} y \right\},
\]
thus the third requirement of Theorem~\ref{Theo32} does not hold.
\end{enumerate}
\end{example}

Strategy-proofness is guaranteed by monotonicity for a group of two teams.
Manipulation in a group of three teams can be presented after considering other tie-breaking rules but this is left to the reader.

Nevertheless, this remains an uninteresting case in practice since discarding the matches against an opponent if there are only two of them can hardly be justified.
Although, for instance, the \href{https://en.wikipedia.org/wiki/2018\%E2\%80\%9319_UEFA_Nations_League_C}{2018-19 UEFA Nations League C} is a tournament containing a group of three teams and a repechage group, some matches are ignored only in groups of four teams for the ranking of third-placed teams in the repechage group.

It is clear that the number of teams in the group, as well as the position of the team transferred to the repechage group, can be modified in Example~\ref{Examp32} without changing the essence of the proof.
Furthermore, if other tie-breaking rules are introduced in Definitions~\ref{Def303} and \ref{Def308}, then $\alpha = \beta$ (a win and a draw have the same worth) or $\beta = \gamma$ (a draw and a loss have the same worth) can be allowed, although this is probably unimportant for most sports.

Based on Figure~\ref{Fig1}, a ``gap'' can be identified between Theorems~\ref{Theo31} and \ref{Theo32} when the set of matches carried over to the repechage group depends on the results, but not only on the group ranking. While such a group-based qualification system is rather theoretical, it reveals the richness of the model contrary to its seemingly restrictive conditions.

\begin{proposition} \label{Prop32}
Let $(\mathcal{T}, \mathcal{G}, \mathcal{R})$ be a monotonic group-based qualification system.
Consider a round-robin group $\left( X^i,S \right)$, $1 \leq i \leq k$, where the following conditions hold:
\begin{itemize}
\item
a team is transferred to the repechage group:
$\left| X^i \cap \mathcal{G}_1(V) \right| > 0$;
\item
there exists a team $x \in X^i \cap \mathcal{G}_1(V)$ such that its set of matches considered in the repechage group depends on the set of group results $V$:
$\mathcal{G}_2(V,x) \neq \mathcal{G}_2(\sigma(V),x)$ for at least one permutation of group results $\sigma$ that satisfies $x \in \left( X^i \cap \mathcal{G}_1(V) \cap \mathcal{G}_1(\sigma(V)) \right)$;
and
\item
there exists a team $x \in X^i \cap \mathcal{G}_1(V)$ such that its set of matches considered in the repechage group does not depend only on the group ranking:
there exist sets of group results $V, \bar{V}$ with $y \succ_{S(v^i)} z$ implying $y \succ_{S(\bar{v}^i)} z$ for all teams $y,z \in X^i$ but $\mathcal{G}_2(V,x) \neq \mathcal{G}_2(\bar{V},x)$.
\end{itemize}
Then it is uncertain whether this group-based qualification system is strategy-proof or not.
\end{proposition}

An example is provided for both cases.

\begin{example} \label{Examp33}
It basically follows the monotonic group-based qualification system $(\mathcal{T}, \mathcal{G}, \mathcal{R})$ of Example~\ref{Examp32} with $\mathcal{G}_2(V,x) = \left\{ y \in X^i: \left| \left\{ z \in X^i: z \succ_{S(v^i)} y \right\} \right| = 0 \right\}$ (Case~\ref{Case1}), thus only the matches played against the top team count in the repechage group.
However, if there exists a team $y \in X^i$ such that $v_1^i(z,y) > v_1^i(z,y)$ and $v_2^i(z,y) < v_1^i(z,y)$ for all $z \in X^i \setminus \{ y \}$ (team $y$ loses all of its matches), then $\mathcal{G}_2(V,x) = \left\{ y \in X^i: \left| \left\{ z \in X^i: z \succ_{S(v^i)} y \right\} \right| = 3 \right\}$, hence only the matches played against the fourth-placed team count in the repechage group.
\end{example}

The monotonic group-based qualification system $(\mathcal{T}, \mathcal{G}, \mathcal{R})$ of Example~\ref{Examp33} is incentive compatible. First, if only the matches played against the top team count in the repechage group (consequently, there exists no team losing all of its matches), then $\mathcal{G}_2(V,x)$ cannot be modified by the second-placed team $x \in X^i$ through exerting a lower effort. Second, if the matches played against the fourth-placed team, which loses all of its matches, count in the repechage group, then the second-placed team $x \in X^i$ has no incentive to manipulate because it carries over the maximal number of points $2 \alpha$ to the repechage group.

\begin{example} \label{Examp34}
It basically follows the monotonic group-based qualification system $(\mathcal{T}, \mathcal{G}, \mathcal{R})$ of Example~\ref{Examp32} with $\mathcal{G}_2(V,x) = X^i \setminus \{ x \} \setminus \left\{ y \in X^i: \left| \left\{ z \in X^i: z \succ_{S(v^i)} y \right\} \right| = 3 \right\}$ (Case~\ref{Case6}), that is, the matches played against the fourth-placed team are discarded in the repechage group.
However, if there exists a team $y \in X^i$ such that $v_1^i(z,y) < v_2^i(z,y)$ and $v_2^i(z,y) > v_1^i(z,y)$ for all $z \in X^i \setminus \{ y \}$ (team $y$ wins all of its matches), then $\mathcal{G}_2(V,x) = \left\{ y \in X^i: \left| \left\{ z \in X^i: z \succ_{S(v^i)} y \right\} \right| \leq 2 \right\}$, hence the matches played against the top team are discarded in the repechage group.
\end{example}

The monotonic group-based qualification system $(\mathcal{T}, \mathcal{G}, \mathcal{R})$ of Example~\ref{Examp34} can be manipulated similarly to Case~\ref{Case6} in the proof of Theorem~\ref{Theo32} because no team has won all of its matches there.

Note that---despite the different setting considered---a similar gray zone emerges in the analysis of tournament systems consisting of multiple round-robin and knock-out tournaments with noncumulative prizes as illustrated by \citet[Example~4]{DagaevSonin2018}.

\begin{corollary} \label{Cor31}
2018 FIFA World Cup qualification tournament (UEFA) is incentive incompatible.
\end{corollary}

\begin{proof}
The scenario presented in Example~\ref{Examp21} shows that team $\text{Bulgaria} = x \in X^1$ can manipulate because there exist sets of group results $V = \left\{ v^1, v^2, \dots ,v^9 \right\}$ and $\bar{V} = \left\{ \bar{v}^1, v^2, \dots ,v^9 \right\}$ such that $\bar{v}^1 = v^1$ except for $\bar{v}_1^1(y,x) = 1 > 0 = v_1^1(y,x)$, where team $\text{Luxembourg} = y \in X^1$ and $\mathcal{R}(V,x) = 0 < 1 = \mathcal{R}(\bar{V},x)$.

Theorem~\ref{Theo32} can also be applied due to Proposition~\ref{Prop31}.
The group-based qualification system of the 2018 FIFA World Cup qualification tournament (UEFA) is monotonic, $\left| X^i \cap \mathcal{G}_1(V) \right| = 1$ and $\mathcal{R}(V,x)$ can be $0$ or $1$ for team $x \in X^i \cap \mathcal{G}_1(V)$ in the repechage group. Furthermore, for each team $x \in X^i \cap \mathcal{G}_1(V)$, $\mathcal{G}_2(V,x)$ depends only on the group ranking, and, finally, $\emptyset \neq \left( \mathcal{G}_2(V,x) \cap \left\{ y \in X^i \setminus \mathcal{G}_1(V): x \succ_{S(v^i)} y \right\} \right) \subset \left\{ y \in X^i \setminus \mathcal{G}_1(V): x \succ_{S(v^i)} y \right\}$ under any set of group results $V = \left\{ v^1, v^2, \dots ,v^k \right\}$.
\end{proof}

Theorem~\ref{Theo32} proves the incentive incompatibility of the \href{https://en.wikipedia.org/wiki/2014_FIFA_World_Cup_qualification_(UEFA)}{2014 FIFA World Cup qualification tournament (UEFA)}, too, which has already been shown in \citet{DagaevSonin2013}.

\begin{remark} \label{Rem33}
The allocation rule $\mathcal{R}$ of the model above does not distinguish between teams advancing to the subsequent round.
This holds if the next round is seeded randomly (like in the \href{https://en.wikipedia.org/wiki/UEFA_Euro_2000_qualifying}{UEFA Euro 2000 qualifying tournament}), or on the basis of an exogenous team ranking that cannot be manipulated by exerting a lower effort in some matches (like FIFA World Rankings used in the case of 2018 FIFA World Cup qualification tournament (UEFA)). However, if, for example, the highest-ranked teams in the repechage group are placed in the first pot before the seeding (like in the \href{https://en.wikipedia.org/wiki/UEFA_Euro_2020_qualifying}{UEFA Euro 2020 qualifying tournament}), then the ranking in the repechage group count, and a successful tanking can reduce the expected strength of the opponents in the next stage.
\end{remark}

\section{Implications} \label{Sec4}

It is known from \citet{DagaevSonin2013} and Corollary~\ref{Cor31}, respectively, that the 2014 and 2018 FIFA World Cup qualification tournaments (UEFA) do not satisfy strategy-proofness.
Further qualification tournaments to the recent FIFA World Cups in the European Zone (World Cup (UEFA)) and UEFA European Championships (UEFA Euro) are analysed from this point of view in Table~\ref{Table5}. 

\begin{sidewaystable}
\centering
\caption[A summary of qualification tournaments to the FIFA World Cups in the European Zone and UEFA European Championships since 1990]{A summary of qualification tournaments to the FIFA World Cups (UEFA) and UEFA European Championships since 1990 \\ \vspace{0.25cm}
\footnotesize{Groups = Number of groups; Teams = Number of teams; Gr / T = Number of groups / Number of teams in the group; Slots = Number of teams qualified; QD = Teams qualified directly; PO = Teams advanced to the play-offs; Discarded matches = Group matches that are discarded in the repechage group; SP = Strategy-proofness (\textcolor{PineGreen}{\ding{52}}: holds; \textcolor{BrickRed}{\ding{55}}: is violated})}
\label{Table5}
\rowcolors{1}{gray!20}{}
\begin{threeparttable}
\begin{tabularx}{1\linewidth}{Lcc ccc cc c} \toprule \hiderowcolors
    Qualification tournament & Groups & Teams & Gr / T & Slots & QD    & PO   & Discarded matches & SP \\ \midrule \showrowcolors
    \href{https://en.wikipedia.org/wiki/1990_FIFA_World_Cup_qualification_(UEFA)}{1990 World Cup (UEFA)} & 7     & 32    & 4/5; 3/4 & 13    & 1st; six 2nd\tnote{1} & ---   & ---   & \textcolor{PineGreen}{\ding{52}} \\
    \href{https://en.wikipedia.org/wiki/1994_FIFA_World_Cup_qualification_(UEFA)}{1994 World Cup (UEFA)} & 6     & 37    & 1/7; 4/6; 1/5\tnote{2} & 12    & 1st; 2nd & ---   & ---   & \textcolor{PineGreen}{\ding{52}} \\
    \href{https://en.wikipedia.org/wiki/1998_FIFA_World_Cup_qualification_(UEFA)}{1998 World Cup (UEFA)} & 9     & 49    & 4/6; 5/5 & 14    & 1st; best 2nd & eight worst 2nd & against 5th and 6th & \textcolor{BrickRed}{\ding{55}} \\
    \href{https://en.wikipedia.org/wiki/2002_FIFA_World_Cup_qualification_(UEFA)}{2002 World Cup (UEFA)} & 9     & 50    & 5/6; 4/5 & 13.5\tnote{3}  & 1st   & all 2nd\tnote{4} & ---   & \textcolor{PineGreen}{\ding{52}} \\
    \href{https://en.wikipedia.org/wiki/2006_FIFA_World_Cup_qualification_(UEFA)}{2006 World Cup (UEFA)} & 8     & 51    & 3/7; 5/6 & 13    & 1st; two best 2nd & six worst 2nd & against 7th & \textcolor{BrickRed}{\ding{55}} \\
    \href{https://en.wikipedia.org/wiki/2010_FIFA_World_Cup_qualification_(UEFA)}{2010 World Cup (UEFA)} & 9     & 53    & 8/6; 1/5 & 13    & 1st   & eight best 2nd & against 6th & \textcolor{BrickRed}{\ding{55}} \\
    \href{https://en.wikipedia.org/wiki/2014_FIFA_World_Cup_qualification_(UEFA)}{2014 World Cup (UEFA)} & 9     & 53    & 8/6; 1/5 & 13    & 1st   & eight best 2nd & against 6th & \textcolor{BrickRed}{\ding{55}} \\
    \href{https://en.wikipedia.org/wiki/2018_FIFA_World_Cup_qualification_(UEFA)}{2018 World Cup (UEFA)} & 9     & 54    & 9/6   & 13    & 1st   & eight best 2nd & against 6th & \textcolor{BrickRed}{\ding{55}} \\ \hline
    \href{https://en.wikipedia.org/wiki/UEFA_Euro_1992_qualifying}{1992 UEFA Euro} & 7     & 33    & 5/5; 2/4 & 7     & 1st   & ---   & ---   & \textcolor{PineGreen}{\ding{52}} \\ 
    \href{https://en.wikipedia.org/wiki/UEFA_Euro_1996_qualifying}{1996 UEFA Euro} & 8     & 47    & 7/6; 1/5 & 15    & 1st; six best 2nd & two worst 2nd & against 5th and 6th & \textcolor{BrickRed}{\ding{55}} \\
    \href{https://en.wikipedia.org/wiki/UEFA_Euro_2000_qualifying}{2000 UEFA Euro} & 9     & 49    & 4/6; 5/5 & 14    & 1st; best 2nd & eight worst 2nd & against 5th and 6th & \textcolor{BrickRed}{\ding{55}} \\
    \href{https://en.wikipedia.org/wiki/UEFA_Euro_2004_qualifying}{2004 UEFA Euro} & 10    & 50    & 10/5  & 15    & 1st   & all 2nd & ---   & \textcolor{PineGreen}{\ding{52}} \\
    \href{https://en.wikipedia.org/wiki/UEFA_Euro_2008_qualifying}{2008 UEFA Euro} & 7     & 50    & 1/8; 6/7 & 14    & 1st; 2nd & ---   & ---   & \textcolor{PineGreen}{\ding{52}} \\
    \href{https://en.wikipedia.org/wiki/UEFA_Euro_2012_qualifying}{2012 UEFA Euro} & 9     & 51    & 6/6; 3/5 & 14    & 1st; best 2nd & eight worst 2nd & against 6th & \textcolor{BrickRed}{\ding{55}} \\
    \href{https://en.wikipedia.org/wiki/UEFA_Euro_2016_qualifying}{2016 UEFA Euro} & 9     & 53    & 8/6; 1/5 & 23    & 1st; 2nd; best 3rd & eight worst 3rd & against 6th & \textcolor{BrickRed}{\ding{55}} \\ \bottomrule
\end{tabularx}

\begin{tablenotes} \footnotesize
	\item[1] The second-placed teams in the four groups containing five teams directly qualified together with the two best second-placed teams in the three groups containing only four teams.
	\item[2] Group 5, originally containing six teams, ended up with five after Yugoslavia was suspended.
	\item[3] One team was advanced to an intercontinental play-off.
	\item[4] The runner-up of Group 2 was drawn randomly for an intercontinental play-off.
\end{tablenotes}
\end{threeparttable}
\end{sidewaystable}

The \href{https://en.wikipedia.org/wiki/UEFA_European_Championship}{UEFA European Championship} has been held every four years since 1960. The \href{https://en.wikipedia.org/wiki/UEFA_European_Championship_qualifying}{qualification systems} for the tournaments between 1960 and 1992 can be described without a repechage group, thus they were strategy-proof due to condition~\ref{Con1} of Theorem~\ref{Theo31}.
The same result provides the incentive compatibility of the \href{https://en.wikipedia.org/wiki/UEFA_Euro_2004_qualifying}{2004} and \href{https://en.wikipedia.org/wiki/UEFA_Euro_2008_qualifying}{2008 qualifying tournament}, where teams from different groups had not to be compared.

The \href{https://en.wikipedia.org/wiki/UEFA_Euro_2020_qualifying}{UEFA Euro 2020 qualifying tournament} is linked with the \href{https://en.wikipedia.org/wiki/2018\%E2\%80\%9319_UEFA_Nations_League}{2018-19 UEFA Nations League} and gives teams a secondary route to qualify for the final tournament through the \href{https://en.wikipedia.org/wiki/UEFA_Euro_2020_qualifying_play-offs}{qualifying play-offs}. This format is not covered by our theoretical model. However, it is also incentive incompatible according to \citet[Proposition~2]{DagaevSonin2018} since the qualifying system essentially consists of two parallel round-robin tournaments, therefore a team might be strictly better off by creating a vacant slot in the qualifying play-offs.\footnote{~There is a debate around tanking in the \href{https://en.wikipedia.org/wiki/UEFA_Nations_League}{UEFA Nations League}, see at \url{https://www.reddit.com/r/soccer/comments/75vrcm/tanking_in_the_uefa_nations_league_wont_benefit/} and \url{http://www.football-rankings.info/2017/07/uefa-nations-league-losing-could.html}. However, it focuses on the observation that dropping to a lower-ranked league could substantially improve the chances of qualifying through the play-offs, an unfair gain only in terms of expected probability \citep{Csato2020b}.}

The first incentive incompatible FIFA World Cup qualifier in the European zone was the \href{https://en.wikipedia.org/wiki/1998_FIFA_World_Cup_qualification_(UEFA)}{1998 FIFA World Cup qualification tournament (UEFA)}.
The \href{https://en.wikipedia.org/wiki/2002_FIFA_World_Cup_qualification_(UEFA)}{2002 FIFA World Cup qualification tournament (UEFA)} again satisfied strategy-proofness due to the lack of a repechage group (see condition \ref{Con1} of Theorem~\ref{Theo31}).

Even though the \href{https://en.wikipedia.org/wiki/1994_FIFA_World_Cup_qualification_(UEFA)}{1990} and  \href{https://en.wikipedia.org/wiki/1994_FIFA_World_Cup_qualification_(UEFA)}{1994 FIFA World Cup qualification  tournaments (UEFA)} were incentive compatible, their fairness remains questionable: the second-placed teams qualified automatically only from certain groups in the first case, and all second-placed teams qualified but group sizes varied in the second case (some qualification tournaments before the 1990 event suffered from the same problem).
It is clear that incentive compatibility is a narrower concept than fairness, see \citet{Guyon2018a} for an examination of the latter through the example of the \href{https://en.wikipedia.org/wiki/UEFA_Euro_2016}{2016 UEFA European Championship}. Hence, one can even say that the organisers sacrificed fairness for the sake of strategy-proofness.

\begin{proposition} \label{Prop41}
Qualification tournaments for the 1996, 2000, 2012 and 2016 UEFA European Championships as well as for the 1998, 2006, 2010, 2014 and 2018 FIFA World Cups in the European Zone were incentive incompatible.
\end{proposition}

\begin{proof}
It immediately follows from Theorem \ref{Theo32} (see Table~\ref{Table5}).
\end{proof}

The list is far from exhaustive.
For example, the elite round of the \href{https://en.wikipedia.org/wiki/2016_UEFA_European_Under-17_Championship_qualification}{2016 UEFA European Under-17 Championship qualification tournament} and the \href{https://en.wikipedia.org/wiki/2017_UEFA_European_Under-21_Championship}{2017 UEFA European Under-21 Championship} should be paid special attention because they may inspire academic research in the future.

It is also worth noting that all incentive incompatible qualifiers of Table~\ref{Table5} discard some matches in the repechage group, which is a necessary---but not sufficient---condition of violating strategy-proofness as Figure~\ref{Fig1} shows.

Proposition~\ref{Prop41} carries a disconcerting message for the administrators of FIFA and UEFA as they could be responsible for a potential scandal occurring in a recent qualification tournament. The example of Section~\ref{Sec2} would be especially disturbing because Luxembourg would have practically no incentive to interfere with the manipulation of Bulgaria and to prevent the elimination of Montenegro. Furthermore, Luxembourg might have even been interested in scoring a goal to obtain the fifth place in the group.
Fortunately, such an outcome did not materialised, and we do not know about any attempt to strategically manipulate these qualification tournaments in the way presented above.

\section{Strategy-proof mechanisms} \label{Sec5}

Nevertheless, the absence of dishonest behaviour in the history of qualification tournaments does not reduce the importance of strategy-proofness in practice, especially if it can be achieved without significant rule changes.
For instance, in the 2018 FIFA World Cup qualification tournament (UEFA), the root of the problem resides in discarding the matches played against the sixth-placed teams in the comparison of runners-up. The greatest pity about this situation is that it could have been straightforward to avoid by UEFA ditching the strange policy of ignoring some group matches since all groups had six teams following the admission of Gibraltar and Kosovo.

Yet the administrators chose not to modify the rules. According to a UEFA News \citep{UEFA2017a}, released on 10 October 2017, \emph{after} the end of group stage: ``[\dots] \emph{the exclusion of results against sixth-placed teams was retained to alleviate any possible imbalance between the qualifying groups caused by the late introductions of Gibraltar and Kosovo}''. While it is respectable to consider some \emph{mathematically unprovable} imbalances between the groups, this decision sacrificed the more clear and crucial requirement of incentive compatibility.

It seems to be necessary to suggest incentive compatible designs in order to argue against the format of recent qualification tournaments.
Denote the number of teams by $n$, and the number of groups by $k$. Let $m \equiv n \mod k$ such that $0 \leq m < k$ and $\ell = (n-m)/k \in \mathbb{Z}$.
We are looking for strategy-proof group-based qualification systems on the basis of Theorem~\ref{Theo31}.

\subsection*{\emph{Mechanism A}: eliminating the repechage group}

The optimal case is to create groups of equal size, like in the \href{https://en.wikipedia.org/wiki/UEFA_Euro_2004_qualifying}{UEFA Euro 2004 qualifying tournament} (see Table~\ref{Table5}). However, it may conflict with divisibility, this solution cannot be followed if $m > 0$.

Another possibility is that all second- or third-placed teams either qualify or advance to the next round regardless of group sizes. For example, the \href{https://en.wikipedia.org/wiki/UEFA_Euro_2008_qualifying}{UEFA Euro 2008 qualifying tournament} was strategy-proof as the top two teams in each group qualified, although a group contained more teams than the others, which may be unfair (see Table~\ref{Table5}).

Mechanism A ensures incentive compatibility due to condition~\ref{Con1} of Theorem~\ref{Theo31}.

Since the application of Mechanism A remains questionable if $m > 0$, three other policies are suggested.

\subsection*{\emph{Mechanism B}: matches to be discarded in the repechage group are independent of group results}

While social choice theory usually wants to avoid the violation of anonymity at all costs, it makes sense to consider this solution because of the seeding procedure. If teams are ranked on an external basis (such as FIFA World Rankings from a given month), there would be $k$ teams in Pot 1, $k$ teams in Pot 2, and so on, until Pot $\ell$ with $k + m$ (as in the \href{https://en.wikipedia.org/wiki/UEFA_Euro_2008_qualifying}{UEFA Euro 2008 qualifying tournament}) or Pot $\ell + 1$ with $m$ (as in the \href{https://en.wikipedia.org/wiki/2018_FIFA_World_Cup_qualification_(UEFA)}{2018 FIFA World Cup qualification tournament (UEFA)}) teams is formed. Remember that $0 < m < k$. Since the last pot is responsible for the difference of group sizes, it could be fair to discard the matches played against the team(s) from the last pot in the repechage group. \\
Nevertheless, a problem may arise when a team from the last pot is transferred to the repechage group due to its unexpectedly good performance in the qualifiers.\footnote{~We are grateful to \emph{D\'enes P\'alv\"olgyi} for spotting this issue. \\
One of the greatest surprises of this type occurred in the \href{https://en.wikipedia.org/wiki/UEFA_Euro_2016_qualifying}{2016 UEFA Euro qualifying tournament} when Greece finished as the last-placed team in Group F despite being drawn from Pot 1.}
The scenario can be immediately addressed by discarding the matches played against the team from the penultimate pot for this particular team, which does not affect strategy-proofness.

Mechanism B provides incentive compatibility due to condition~\ref{Con2} of Theorem~\ref{Theo31}.

For the 2018 FIFA World Cup qualification tournament (UEFA), this regime means fixing in advance that matches played against the teams in Pot 6 (Luxembourg, Andorra, San Marino, Georgia, Kazakhstan, Malta, Liechtenstein, as well as the lately introduced Gibraltar and Kosovo) are ignored in the comparison of the runners-up. Since only Luxembourg and Georgia obtained a better (the fifth) position in the qualification tournament, this policy does not make much difference in practice. However, while it is attractive from a theoretical point of view, ethical concerns may arise.

\subsection*{\emph{Mechanism C}: all or no matches played against higher- and lower-ranked teams are considered in the repechage group}

If group sizes vary because of $m > 0$, and the $p$th-placed teams of the groups are considered in the repechage group, then only their matches played against higher-ranked teams can count in the repechage group as the cardinality of the set of teams that are ranked lower than the $p$th-placed teams in the groups is different.
It means no problem when:
\begin{enumerate}[label=\alph*)]
\item
The qualification tournament is centred around relegation: in the \href{https://en.wikipedia.org/wiki/2018\%E2\%80\%9319_UEFA_Nations_League_C}{2018-19 UEFA Nations League C}, $15$ teams are divided into one group of three teams and three groups of four teams each such that the three fourth-placed and the worse third-placed teams are relegated to League D, at least according to the original intention of the UEFA. In the comparison of the third-placed teams, only the matches played against the first- and second-placed teams in the group count. Since there is at most one team ranked lower than the team transferred to the repechage group in each group, they cannot manipulate by changing their set of matches to be ignored.
\item
A high proportion of participating teams qualify: in the \href{https://en.wikipedia.org/wiki/2018_European_Men's_Handball_Championship_qualification\#Qualification_Phase_2}{2018 European Men's Handball Championship Qualification Phase 2}, the contestants were split into seven groups of four teams each such that the top two ranked teams from each group and the best third-placed team qualified for the final tournament. In the comparison of third-placed teams, only the matches played against the teams ranked first and second in the group was considered.
\end{enumerate}

Mechanism C provides incentive compatibility due to condition~\ref{Con3} of Theorem~\ref{Theo31}.

This policy has guaranteed the strategy-proofness of the \href{https://en.wikipedia.org/wiki/UEFA_Euro_2016}{2016 UEFA European Championship} with uniform group sizes \citep{Guyon2018a}, too, where third-placed teams were ranked on the basis of all group matches.

\subsection*{\emph{Mechanism D}: new seeding policy}

The idea behind mechanism C can be applied with a slight modification of the seeding procedure, by using a bottom-up design of the pots instead of the usual top-down approach. First, the worst $k$ teams form Pot $\ell+1$, the next $k$ teams Pot $\ell$, and so on until Pot 1 with $m$ teams is formed, where $0 < m < k$. Then $m$ groups are created by drawing a team from each pot, and $k-m$ groups are created by drawing a team from each pot except for Pot 1.
The top $p$ teams qualify from the $m$ groups of size $\ell+1$, and the top $p-1$ teams qualify from the $k-m$ groups of size $\ell$. The repechage group consists of the $(p+1)$th-placed teams from the $m$ groups of size $\ell+1$, and the $p$th-placed teams from the $k-m$ groups of size $\ell$, where they are compared on the basis of their matches played against the $\ell - p$ teams that are ranked lower in their groups. Consequently, the matches played against the already qualified teams are discarded in the repechage group, which seems to be a reasonable policy.

Mechanism D provides incentive compatibility due to condition~\ref{Con3} of Theorem~\ref{Theo31}.

\subsection*{Further solutions}

Some other designs can guarantee strategy-proofness. For example, the format of the \href{https://en.wikipedia.org/wiki/1990_FIFA_World_Cup_qualification_(UEFA)}{1990 FIFA World Cup qualification tournament (UEFA)} can be followed without discarding any matches. It means that: (1) the directly qualified runners-up or third-placed teams are chosen from the groups containing more teams, or (2) the eliminated runners-up or third-placed teams are chosen from the groups containing more teams. However, this solution seems to be dishonest as discussed in Section~\ref{Sec4}.

It is also possible to organise a preliminary round for lower-ranked teams, either a round-robin such as in the \href{https://en.wikipedia.org/wiki/2018_FIVB_Volleyball_Men\%27s_World_Championship_qualification_(CEV)}{CEV qualification tournament for the 2018 FIVB Volleyball Men's World Championship}, or a two-leg playoff, similarly to the \href{https://en.wikipedia.org/wiki/2018_FIFA_World_Cup_qualification_\%E2\%80\%93_AFC_First_Round}{2018 FIFA World Cup qualification tournament---AFC First Round}.
Note that the match calendar is less crowded for lower-ranked teams because they normally do not qualify for major tournaments.

Another solution might be reducing the number of contestants to achieve equal group sizes. For instance, the weakest teams (e.g. Gibraltar, Liechtenstein, San Marino, etc.) can be placed into a special group, where they play against each other without the possibility of direct qualification. The winner of this group may advance to a play-off with the best runner-up or third-placed team. Besides excluding manipulation, this policy has the further benefit of giving a chance for lower-ranked national teams, mainly composed of amateur players, to compete against each other and enjoy more success than scoring some lucky goals against professional sportsmen.\footnote{~This idea may have inspired the novel tournament of \href{https://en.wikipedia.org/wiki/UEFA_Nations_League}{UEFA Nations League}, especially its \href{https://en.wikipedia.org/wiki/2020\%E2\%80\%9321_UEFA_Nations_League}{2020/21 season}, where 48 teams play in three leagues containing four groups of four teams each, while the remaining seven UEFA national teams form the lowest-ranked League D.
A similar mechanism is used in the \href{https://en.wikipedia.org/wiki/2022_FIFA_World_Cup_qualification_(CONCACAF)}{North, Central American and Caribbean section of the 2022 FIFA World Cup qualification}: the fourth-placed team in the top-seeded \href{https://en.wikipedia.org/wiki/Hexagonal_(CONCACAF)}{Hexagonal} group of six teams plays against the winner of the tournament organised for the remaining 29 teams to obtain the inter-confederation play-off slot of the Confederation of North, Central American and Caribbean Association Football (CONCACAF).}

\subsection*{Assessment}

To summarise, it is clear that mechanism A is perfect in the case of $m=0$ but its fairness remains questionable if $n$ is not divisible by $k$.

Mechanism C leads to discarding a high proportion of matches in the repechage group for the qualification tournaments presented in Table~\ref{Table5}, which may result in increased randomness. It means a problem in qualification for a tournament with huge financial issues at stake---for example, all teams were guaranteed at least USD 9.5 million each for their participation in the 2018 FIFA World Cup \citep{FIFA2017b}.

Consequently, we propose to apply mechanisms B or D if uniform group size cannot be achieved. Both are general, that is, they can be directly applied for arbitrary values of $k$ and $n$, and they coincide with mechanism A if $m=0$.

Finally, note that strategy-proofness cannot be guaranteed without some costs.
For example, Mechanism B requires fixing the team within a group as the one against which the results will be ignored in the repechage group, therefore some opponents (mainly the highest-ranked ones) may reduce their effort in those matches. However, no results are discarded in the group rankings, and this drawback more or less applies to the original incentive incompatible design, too. Thus, in our opinion, the potential reduction in efforts compared to the currently used regime is a small price for eliminating the opportunity that a team might be strictly better off by losing.

\section{Conclusions} \label{Sec6}

Tournament organisers may face unpleasant situations when they miss considering strategy-proofness as such a tanking of a team can easily undermine the integrity of the sport.
We have demonstrated that qualification tournaments for some recent FIFA World Cups and UEFA European Championships were incentive incompatible, and suggested some alternative incentive compatible designs. One of the proposals also yields the important lesson that violating anonymity is not necessarily harmful in mechanism design.

There are at least three possible directions for future research.
First, several distinct tournament formats can be studied from the perspective of strategy-proofness.
Second, the current theory-oriented investigation can be supplemented by estimating the probability of manipulation with the use of historical and Monte-Carlo simulated data. 
Finally, the solutions guaranteeing incentive compatibility can be compared with respect to other criteria.

Hopefully, this paper reinforces our view that the scientific community and the sports industry should work more closely together in analysing the effects of potential rules and, especially, rule changes, even before they are implemented. For example, the governing bodies of major sports may invite academics to identify possible loopholes in proposed regulations to prevent scandals in the future.


\section*{Acknowledgements}
\addcontentsline{toc}{section}{Acknowledgements}
\noindent
We are grateful to \emph{Dmitry Dagaev}, \emph{Tam\'as Halm} and \emph{L\'aszl\'o \'A. K\'oczy} for useful advice. \\
Nine anonymous reviewers provided valuable comments and suggestions on earlier drafts. \\
We are indebted to the \href{https://en.wikipedia.org/wiki/Wikipedia_community}{Wikipedia community} for contributing to our research by collecting and structuring information on the sports tournaments discussed. \\
The research was supported by the MTA Premium Postdoctoral Research Program grant PPD2019-9/2019.

\bibliographystyle{apalike}
\bibliography{All_references}

\end{document}